\documentclass{article}

\usepackage{fullpage}
\usepackage{booktabs} 
\usepackage{hhline}
\usepackage{enumitem}
\usepackage{listings, fancyvrb, multicol}
\usepackage{latexsym,amsmath,amsfonts,amssymb,stmaryrd,mathtools}
\usepackage{amsthm}
\usepackage{xcolor}
\usepackage{algorithm, algorithmicx}
\usepackage[noend]{algpseudocode}
\usepackage{graphicx}
\usepackage{hyperref}
\usepackage{wrapfig}
\usepackage{multirow}
\usepackage{times}
\usepackage{comment}
\usepackage{enumitem}
\usepackage{xspace}
\usepackage{pifont}
\usepackage{parcolumns}
\usepackage[rightcaption]{sidecap}
\usepackage{balance}

\newcommand{\myparagraph}[1]{\noindent\emph{\textbf{#1\,}}}

\usepackage{environ}

\newif\ifappendix
\NewEnviron{maybeappendix}[1]
{\ifappendix
	\expandafter\global\expandafter\let\csname putmaybeappendix#1\endcsname\BODY%
	\else
	\expandafter\newcommand\csname putmaybeappendix#1\endcsname{}\BODY%
	\fi}
\newcommand{\putmaybeappendix}[1]{\csname putmaybeappendix#1\endcsname}

\newif\iffull
\NewEnviron{fullversion}[1]
{\iffull
#1
	\else
	\fi}

\fullfalse

\newtheorem{theorem}{Theorem}[section]
\newtheorem{lemma}[theorem]{Lemma}

\newtheorem{definition}[theorem]{Definition}
\newtheorem{observation}[theorem]{Observation}

\newcommand{\VerMain}{Version Maintenance}

\newcommand{\set}[1][]{\texttt{set}$_{#1}$}
\newcommand{\acquire}[1][]{\texttt{acquire}$_{#1}$}
\newcommand{\release}[1][]{\texttt{release}$_{#1}$}
\newcommand{\getdata}[1][]{\texttt{getData}}
\newcommand{\currentv}{\texttt{V}}

\newcommand{\rlock}{\texttt{read\_lock}}
\newcommand{\runlock}{\texttt{read\_unlock}}
\newcommand{\rsync}{\texttt{\texttt{synchronize}}}

\newcommand{\vptr}{version}
\newcommand{\collect}{\texttt{collect}}
\newcommand{\outputt}{\texttt{output}}
\newcommand{\op}{operation}

\newcommand{\tuple}{tuple}
\newcommand{\activet}{active}
\newcommand{\PSWF}{PSWF}

\newcommand{\true}{true}
\newcommand{\false}{false}

\newcommand{\hide}[1]{} 
\newcommand{\vm}{M}

\def\phasetwo{run phase\xspace}
\def\phasethree{exit phase\xspace}

\appendixtrue

\renewcommand{\thefootnote}{\fnsymbol{footnote}}

\begin{document}
	
\title{
Multiversion Concurrency with Bounded Delay\\ and Precise Garbage Collection\footnote{This paper is the full version of a paper published in the ACM Symposium on Parallelism in Algorithms and Architectures (SPAA), 2019. The conference version of this paper \cite{singlewriter} can be found at \url{https://doi.org/10.1145/3323165.3323185}.
}}

\author{Naama Ben-David \quad Guy E. Blelloch \quad Yihan Sun \quad Yuanhao Wei \\ \{nbendavi, guyb, yihans, yuanhao1\}@cs.cmu.edu\\Carnegie Mellon University}

	\maketitle

\begin{abstract}

  In this paper we are interested in bounding the number of
  instructions taken to process transactions. The main result is a
  multiversion transactional system that supports constant delay
  (extra instructions beyond running in isolation) for all read-only
  transactions, delay equal to the
  number of processes for writing transactions that are not concurrent
  with other writers, and lock-freedom for concurrent writers. The
  system supports precise garbage collection in that versions are
  identified for collection as soon as the last transaction releases
  them. As far as we know these are first results that bound delays
  for multiple readers and even a single writer. The approach is
  particularly useful in situations where read-transactions dominate
  write transactions, or where write transactions come in as streams
  or batches and can be processed by a single writer (possibly in
  parallel).

  The approach is based on using functional data structures to
  support multiple versions, and an efficient solution to the Version
  Maintenance (VM) problem for acquiring, updating and releasing versions.
  Our solution to the VM problem is precise, safe and wait free (\PSWF).

We experimentally validate our approach by applying it to balanced tree data structure for maintaining ordered maps.  
We test the transactional system using multiple algorithms for the VM problem, including our \PSWF{} VM algorithm, and implementations with weaker guarantees based on epochs, hazard pointers, and read-copy-update. 
To evaluate the functional data structure for concurrency and multi-versioning, we implement batched updates for functional tree structures and compare the performance with state-of-the-art concurrent data structures for balanced trees.
The experiments indicate our approach works well in practice over a broad set of criteria.
\end{abstract}

\renewcommand*{\thefootnote}{\arabic{footnote}}
	
	\pagenumbering{arabic}
	\section{Introduction}

Consider a sequential computation that takes $\tau$ instructions
(time) to run.  If the computation is run by some system atomically as
a transaction\footnote{Throughout we use ``transaction'' to mean the
  traditional sense of a sequence of instructions that appear to take
  place atomically at some point during their execution (strictly
  serializable)~\cite{papadimitriou1979serializability}, and not to
  mean a specific implementation technique such as transactional
  memory.}  concurrently with other transactions that share data, we
would expect it would take more time to complete.  This can be both
due to the overhead of the transactional system, and due to inherent
dependences among the transactions, forcing the system to wait for
another to complete.
In this paper we are interested in bounding the extra time.  We say
the sequential computation has $O(\delta)$ delay if its transaction
completes in $O(\tau + \delta)$ time.

In general, it is impossible to bound the delay by better than
$O({\tau \times p})$, even ignoring overheads, since for a set of $p$
transactions with equal $\tau$, the dependences between them might
require that they fully sequentialize.  For example, consider an
integer variable $x$ stored in a shared location, an arbitrary unknown
function $f$, and the transaction $x = f(x)$.  If the same transaction
is applied concurrently on $p$ processes, the transactions need to
fully sequentialize for correctness.  Hence if $f$ takes $\tau$
time on its own, and if all processes are working at the same
rate, one transaction will have to wait for at least $\tau \times p$
time to complete.

When most transactions are read-only, however, the prognosis is
significantly better.  In particular, read-only transactions (readers)
can in principle proceed with constant delay and without delaying any
writing transactions (writers), since they do not modify any memory, and hence other transactions do not depend on them.
This can be very useful in workloads
dominated by readers.  Several approaches try
to take advantage of this.  Read-copy-update (RCU)~\cite{McKenney98} allows for
an arbitrary number of readers to proceed with constant delay, and has
become a core idiom widely used in Linux and other operating
systems~\cite{McKenney01}.  In RCU, however, readers can arbitrarily delay
(block) a writer, since a writer cannot proceed until all readers have
exited their transaction.  This is particularly problematic if some
readers take significant time, fault, or sleep~\cite{Matveev15}.  Indeed RCU in
Linux is used in a context in which the readers are short and cannot
be interrupted.  With multi-versioning~\cite{Reed78,BG83,papadimitriou1984concurrency,perelman2010maintaining,Kumar14,neumann2015fast},
on the other hand, not only can readers proceed with constant delay,
but in principle, they can avoid delaying any writers---a writer can
update a new version while readers continue working on old versions.
Therefore a single writer and any number of readers should all be able
to proceed without delay (multiple writers can still delay each other).

Multi-versioning, however, has some significant implementation issues
that can make the ``in principle'' difficult to achieve in ``theory''
or ``practice''.  One is that memory can become an issue due to
maintaining old versions, possibly leading to unbounded
memory usage.  Ideally one would like to reclaim the memory used by a
version as soon as the last transaction using it finishes.  Some
recent work has studied such bounds in memory usage~\cite{perelman2010maintaining}.
Although their results ensure
readers are not blocked and do not block writers, they do not
bound delay.  Another problem arises in the most widely used
implementation of multi-versioning, which involves keeping a version
list for every object~\cite{Reed78,BG83,papadimitriou1984concurrency,Kumar14}.  The
problem is that these lists need to be traversed to find the relevant
version, which causes extra delay for reads.  The delay is not just
a constant, but can be asymptotic in the number of versions.  We
know of no multi-versioned system that can both bound the delay and
ensure memory usage bounds, even when only a single writer is allowed at any time.

In this paper, we develop strong asymptotic bounds on the
delay for transactions while also ensuring bounded memory.
We show what we believe are the first non-trivial cost bounds for transactions
with multi-versioning.  In particular,
for $p$ processes we describe a system with the following properties:

\begin{itemize}[noitemsep,topsep=0pt]
  \item
Read transaction are \emph{delay-free}---i.e., if they take $\tau$ time
(instructions) in the original code, they take $O(\tau)$ time
in the transactional version, from invocation to response.
  \item
A single write transaction (without other concurrent write
transactions) has $O(p)$ delay from invocation to response (i.e. when
the result is visible).
\item
  Multiple concurrent write transactions are \emph{lock-free},
  although a successful write will abort other active writers.
\item The garbage collector is \emph{precise} in that the memory
  associated with any version (except the latest) is collected as soon
  as the last transaction that holds it completes.  Furthermore, the cost of
  the collection is linear in the amount of garbage collected.
\item
  A single writer transaction along with read transactions (not including the
  garbage collection) have constant amortized memory contention.
\end{itemize}
These properties are true for arbitrarily long transactions that
access an arbitrary memory footprint for read-only transactions, and
update an arbitrary number of locations for writing
transactions.

Our approach is particularly useful in read-dominated workloads in
which a single (or very few) writer does updates, or in workloads in
which concurrent writes can be batched into single transactions in the
style of flat-combining~\cite{hendler2010flat}, and then applied by a
single writer.  As with flat-combining, batching gives up on the
wait-freedom of writes, however it allows the writes to
run in parallel potentially getting high throughput.   We study this
in our experiments.

To achieve these bounds we require that programs are implemented using
purely functional data
structures~\cite{Okasaki98,BenAmram95,KT96,Pippenger97}.  Such data
structures are widely used in languages such as F\#, Scala, OCaml,
Haskell, JavaScript, Julia, and Clojure, and date back to the 1950s
with Lisp~\cite{McCarthy60}.  They are also used in various database
systems~\cite{LMDB,couchdb,hyder,innodb}, and sometimes referred to as
copy-on-write~\cite{minuet,becker1996asymptotically}.
On
updates, the path to the update is copied.  Most standard data types
can be implemented efficiently (asymptotically) in the functional
setting, including balanced trees, queues, stacks and priority queues.
Since functional data structures are persistent (immutable), they are
naturally multi-versioned.  Applying an update leaves the old version
while creating a new version.  The version can be accessed via a
pointer to the root, and hence each version is simply a pointer to a
data structure.  The cost of traversing the structures is unaffected
by the versions (unlike version lists).  However, the problem remains
of how to ensure precise garbage collection.

\lstset{basicstyle=\footnotesize\ttfamily, tabsize=1, escapeinside={@}{@}}
\lstset{literate={<<}{{$\langle$}}1  {>>}{{$\rangle$}}1}
\lstset{language=C, morekeywords={CAS,commit,empty,local,job,taken,entry,GOTO}}
\lstset{xleftmargin=5.0ex, numbers=left, numberblanklines=false, frame=single}
\begin{figure}\small
\centering Read Transaction
\begin{lstlisting}
v = acquire(k);
user_code(v);	
// response
versions = release(k);
for (v in versions) collect(v);
\end{lstlisting}
\centering Write Transaction
\begin{lstlisting}
v = acquire(k);
newv = user_code(v);
flag = set(newv);				
// response if successful--- update visible here
versions = release(k);	
for (v in versions) collect(v);
if (!flag) collect(newv) and retry or abort
\end{lstlisting}
\vspace*{-.1in}
\caption{Read and Write transactions with \acquire, \set, and \release.
  $k$ is the process ID.}
\label{alg:framework}
\end{figure}

\hide{
\begin{figure*}\small
\begin{minipage}{.47\textwidth}
\begin{lstlisting}
v = acquire(k);
user_code(v);	
// response
versions = release(k);
for (v in versions) collect(v);
\end{lstlisting}
\centering Read Transaction
\end{minipage}\hspace*{.1in}
\begin{minipage}{0.49\textwidth}
\begin{lstlisting}
v = acquire(k);
newv = user_code(v);
flag = set(newv);				
// response if successful--- update visible here
versions = release(k);	
for (v in versions) collect(v);
if (!flag) collect(newv) and retry or abort
\end{lstlisting}
\centering Write Transaction
\end{minipage}
\vspace*{-.1in}
\caption{Read and Write transactions with \acquire, \set, and \release.
  $k$ is the process ID.}
\label{alg:framework}
\end{figure*}
}

For the purpose of garbage collection, we introduce the version
maintenance (VM) problem.  The problem is to implement a linearizable
object with three operations: \acquire, \release{} and \set.  The
\acquire{} operation returns a handle to the most recent version,
in a way that ensures it cannot be collected.  The \set{} operation updates the
current version to a new pointer, returning whether it succeeded or
failed.  The \release{} operation indicates that the currently
acquired version is no longer needed by the process, potentially making it available to be collected.
It returns a list of versions that can be collected---i.e.,
for which no other process has acquired it and not released it.
Only one version can be acquired on any process at any time, i.e. the
current version must be released before a new one is acquired.  In the
\emph{precise} VM problem, the release will return a singleton list
precisely when the process is the last to release its version, and an
empty list otherwise.  We give a solution to the precise version.

The VM object can be used to implement read-only and writing
transactions as shown in Figure~\ref{alg:framework}.
The read transaction is effectively done after step 2 (response
could be sent to a client), and the rest is a cleanup phase for the
purpose of GC.   Similarly, writing transactions are
done after step 3, at which point the result is visible to other transactions.
After the release, any garbage can be traced from the
released pointers and collected in work linear in the amount of
garbage collected using a standard reference counting collector.

\begin{table}
  \centering
  \small
  \begin{center}
  \begin{tabular}{c|rc@{ }c@{ }|@{}c@{}}
    \hhline{-----}
    &\multicolumn{3}{c|}{\textbf{Time Bound}} & \textbf{Properties}\\
\hhline{-----}
    & \multicolumn{3}{c@{ }|}{Thm. \ref{thm:waitfree-step} and \ref{thm:contention}} &\multicolumn{1}{@{}c@{}}{Thm. \ref{thm:waitfree-lin}}\\
\cline{2-5}
    \textbf{} & &    \textbf{Time}&\textbf{Contention} & No abort and wait-free\\
\cline{2-4}
    \textbf{VM}&\textbf{\acquire{}} & $O(1)$& $O(1)$& for readers and one writer,\\
    \textbf{}&\textbf{\release{}} & $O(P)$& $O(P)$&linearizable\\
    &\textbf{\set{}} & $O(P)$& $O(P)$&\\
     \hhline{-----}
    &\multicolumn{3}{c|}{Thm. \ref{thm:transaction-delay}, \ref{thm:transContention} and \ref{thm:collect-correct}}& Thm. \ref{thm:serial} and \ref{thm:safe-precise}\\

    \cline{2-5}
 \textbf{In} & \textbf{Reader}& \multicolumn{2}{l|}{delay-free} & No abort and wait-free \\
 \textbf{All}   & \textbf{Writer} &\multicolumn{2}{l|}{$O(P)$-delay}&  for readers and one writer,\\
    &\textbf{GC} & \multicolumn{2}{l|}{$O(S+1)$ time}&serializable, safe and precise GC\\
\hhline{-----}
  \end{tabular}
  \end{center}
  \caption{\small The time bounds and properties guaranteed by our algorithms and the corresponding theorems in this paper. ``VM'' means the \VerMain{} problem. $P$ is the number of processes. The contention bounds are amortized. In GC, $S$ is the number of tuples that were freed. ``Delay'' is defined in Section 2. Safe and precise GC are defined in Section \ref{sec:gc}. }\label{tab:allresults}
  
\end{table}

We describe a wait-free algorithm for the precise VM problem, which we refer to as the \PSWF{} algorithm.
It supports the acquire with $O(1)$ delay, and
set and release with $O(p)$ delay.  A read-only transaction only costs
the delay of an acquire (constant), followed by the cost of the
transaction itself, which is unaffected by the multi-versioning (e.g.,
a search in a balanced tree will take $O(\log n)$ time).
In our implementation, the \set{} can only fail
if a concurrent writer has succeeded between its \acquire{} and
\set{}.
Therefore a non-conflict writing transaction takes effect in the
time of the transaction itself plus the cost of the
acquire and set, which is $O(p)$ time (for the \set{}).
We also consider the memory contention of the three operations.
The costs and properties are summarized in Table~\ref{tab:allresults}.


We finish by describing some experiments for both the VM algorithms and the functional data structures.
We test the transactional system using multiple VM algorithms in our framework, including our \PSWF{} algorithm, and implementations with weaker guarantees based on epochs and hazard pointers. Experiments show that our \PSWF{} algorithm on average uses 60\%-90\% less memory for versions than the other two implementations because of precise garbage collection. Our algorithm also achieves comparable throughput to the other two implementations.

To evaluate the functional data structure for concurrency and multi-versioning, we implement batched updates for functional trees and compare the performance with existing concurrent data structures. 
Experiments show that in the tested workloads with mixed reads and updates, using functional data structures with batching can outperform concurrent data structures by more than 20\%.

	\section{Preliminaries}
\label{sec:prelim}

We consider \emph{asynchronous shared memory} with $P$ processes.
%
%
Each process $p$ follows a deterministic sequential protocol composed of \emph{primitive \op{}s} (read, write, or compare-and-swap) to implement an object.
We define \emph{objects, operations} and \emph{histories} in the standard way~\cite{HS08}. We consider \emph{linearizability} as our correctness criterion \cite{herlihy1990linearizability,herlihy1991wait}.
An \emph{adversarial scheduler} determines the order of the invocations and responses in a history.
We refer to some point in a history as a \emph{configuration}. We define the \emph{time complexity} of an operation to be the number of instructions (both local and shared) that it performs. Note that this is different from the standard notion of \emph{step complexity} which only counts access to shared variables.

%

\myparagraph{Transactions.}
We consider two types of transactions: \emph{read-only} and \emph{write}. Each transaction has an \emph{invocation}, a \emph{response}, and a \emph{completion}, in that order. A transaction is considered \emph{active} between its invocation and response, and \emph{live} between its invocation and completion.
Intuitively, the transaction is executed between its invocation and response, and does some extra `clean-up' between its response and its completion.
We require that transactions be strictly serializable, meaning that each transaction appears to take effect at some point during its active interval.
We refer to a write transaction as \emph{single-writer} if no other write transaction is live while it is live.

\myparagraph{Delay.}  We say that the \emph{time}, of a computation
(or algorithm) on a single process is the number of instruction steps
that the computation executes, including all local and shared
instructions.  We say that the \emph{user instructions} of a
transaction are the instructions that would be run in a sequential
setting using regular reads and writes.  We want to simulate these
instructions in a way that the transaction appears atomically in the
concurrent setting.  Consider a transaction that executes user code
that consists of $m$ user instructions.  Such a simulation has
\emph{delay $d$} if the active interval takes $O(d + m)$ time, similarly to~\cite{bendavid2019delay}.
A transaction is \emph{delay-free} if the delay is constant
(or zero).  The $O(d+m)$ bound includes
all instructions needed to ensure strict serializability, and the
big-O is independent of the number of processes, the number of
versions, or the actions of any other concurrent processes.  In a
traditional multiversion system, for example, the bound needs to included
the possibly large number of instructions needed to traverse a version
list.

\myparagraph{Contention.}
We say that the amount of contention experienced by a single
shared-memory \op{} $i$ in a history $H$ is the number of \emph{responses} to modifying \op{}s on the same location that occur between $i$'s invocation and response in $H$.
Note that this is not exactly the definition presented in any previous paper, but it is strictly stronger (implies more contention) than both the definition of Ben-David and Blelloch \cite{ben2017analyzing} and the definition of Fich \textit{et al.} \cite{fich2005linear}. Therefore, the contention results in this paper hold under the other models as well.

\begin{figure}
     \begin{minipage}{.4\columnwidth}
       \framebox{\includegraphics[width=0.8\columnwidth]{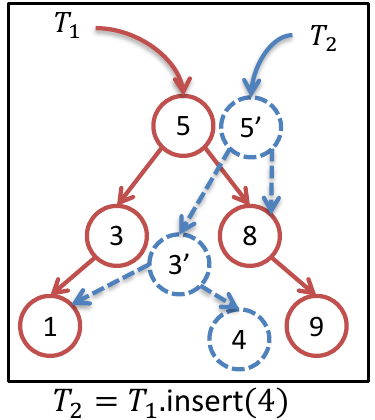}}
     \end{minipage}
     \begin{minipage}{.5\columnwidth}
       \caption{{\small An example of the insert function under PLM using path copying. The output $T_2$ is represented by the root pointer at $5'$, while the input $T_1$ can still be represented by the original root pointer at $5$.
}\label{fig:persistence}}
     \end{minipage}

\end{figure}

\myparagraph{Functional Data Structures.}  We assume that the memory
shared by transactions is based on purely functional (mutation-free)
data structures.  This can be abstracted as the \emph{pure LISP
  machine}~\cite{Okasaki98,BenAmram95,Pippenger97} (PLM), which, like
the random access machine model (RAM), has some constant number of
registers.  However, the only instructions for manipulating memory,
are (1) a \texttt{tuple}$(v_1,\ldots,v_l)$ instruction, which takes
$l$ registers (for some small constant $l$) and creates a tuple in
memory containing their values, and (2) a \texttt{nth}$(t, i)$
instruction, which, given a pointer $t$ to a tuple and an integer $i$
(both in registers), returns the i-th element in this tuple.  Values
in the registers and tuples are either primitive, or a pointer to
another tuple.  There is no instruction for modifying a tuple.
Changing a data structure using PLM instructions are done via
\emph{path copying}, meaning that to change a node, its ancestors in
the data structure must be copied into new tuples, but the remainder
of the data remains untouched.
Using
PLM instructions, one can create a DAG in memory,
which we refer to as the \emph{memory graph}.
A special and commonly-used case for the memory graph is a tree structure.

We define the \emph{version root} as a pointer to a tuple, such that the data
reachable from this tuple constitutes the state that is visible to a
transaction.  Then each update on version $v$ yields a new version by path-copying starting from
the version root of $v$, and the new copied root provides the view to the new version.   An example of using path-copying
to insert a value into a binary tree memory graph is shown in Figure
\ref{fig:persistence}.
In our framework, every transaction $t$ acquires exactly one version $V(t)$.
If $t$ has not yet
determined its version at configuration $C$, then $V_C(t) = null$
until it does.
We use the version roots as the data pointers in the \VerMain{} problem.

\myparagraph{Garbage Collection.}
We assume all tuples are allocated at their tuple
instruction, and freed by a \texttt{free} instruction in the GC.  The \emph{allocated space} consists of all tuples that are allocated and not yet freed.
For a set of transactions $T$, let
$R(T)$, or the \emph{reachable space} for $T$ in configuration $C$, be the set of tuples
that are reachable in the memory graph from their corresponding
version roots, plus the current version $c$, i.e. the \tuple{s} reachable from $\{ V(t) \vert t \in T \} \cup \{c\}$.
We say that a \tuple{}  $u$ \emph{belongs} to a version $v$ if $u$ is reachable 
from $v$'s version root. Note that $u$ can belong to multiple versions. We define a precise and a safe GC, respectively, as follows.

\begin{definition}\label{def:precisegc}
	A garbage collection is \emph{precise} if
	the allocated space at any point in the user history is a subset of the
	reachable space $R(T)$ from the set of live transactions $T$.
\end{definition}

\begin{definition}\label{def:safegc}
  A garbage collection is \emph{safe} if the allocated space is always a superset of the reachable space from the active transactions.
\end{definition}

Roughly speaking, precise GC means
to free any out-of-date tuples in time, and safe GC means not to free any tuples that are currently used by a transaction.

\hide{
\myparagraph{Garbage Collection.}
\emph{Reference counting}~\cite{Collins60,Jones11} is a common approach for enabling safe garbage collection. The idea is that each object maintains a count of the number of references to it, and when it reaches $0$, it is safe to collect.

Intuitively, we say that garbage collection is \emph{precise} if
at any point in the user history, any \tuple{} that are not used by the \activet{} transactions $T$ are freed.
This says that all out-of-date \tuple{s} are freed in time.
We say that garbage collection is \emph{safe} if none of the \tuple{s} used in the set of transactions that are not yet in their \phasethree{} are freed.
We give the formal definition of \emph{precise} and \emph{safe} garbage collection in Section \ref{sec:functional}.
}

	\section{The Version Maintenance Problem}\label{sec:vermain}

In our transaction framework, we abstract what we need for the purpose of maintaining versions as the
\emph{\VerMain{}} problem, which tackles entering and exiting the transactions (see Figure \ref{alg:framework}).

The \VerMain{} problem, or \VerMain{} object, supports three operations: \set{}, \acquire{}, and \release{}.
%
At a high level, the \acquire{} operation returns a version for the process to use and \release{} is called when the process finishes using the version.
New versions are created by \set{} operations.
All three operations take as input an integer $k$ that represents the id of the process that calls the operation. The \set{} operation in addition takes in a pointer to the new version that it should commit, and returns a flag indicating whether or not it succeeded.

We refer to the pointer to a version as the \emph{data pointer}.
More formally, if $d$ is a pointer to data, \set($d$), if successful, creates a new version with  pointer $d$ and sets it as the \emph{current version}, i.e.,
\begin{definition}\label{def:current}
  The \emph{current version} is defined as the version set by the most recent successful \set{} operation.
\end{definition}

The operations are intended to be used in a specific order: an \acquire($k$) should be followed by a \release($k$), with at most one \set($k$, $d$) in between, where $d$ is a pointer to a new version. If this order is not followed for each $k$, then the operations may behave arbitrarily; that is, we do not specify a `correct' behavior for the operations of a \VerMain{} object $O$ in an execution once any operations are called out of this order on $O$.

We define the liveness of a version $v$ as follows.
\begin{definition}\label{def:live}
	A version $v$ is \emph{live} at time $t$ if it is the current version at $t$, or if $\exists k$, s.t. an \acquire($k$) operation $A$ has returned $v$ but no \release($k$) has completed after $A$ and before $t$.
\end{definition}
We note that a version is live while a transaction using that version
is active.   The transaction itself can remain live after its version is
dead, while it garbage collects.

The following is the sequential specification of these operations assuming that they are called in the correct order (\acquire-\release{} or \acquire-\set-\release{} for each id $k$).

\begin{itemize}[noitemsep,topsep=0pt]
	\item \texttt{data* acquire(int k):} Returns the current version.
	\item \texttt{data** release(int k):} 
	Returns a (possibly empty) list of versions that are no longer live. No version can be returned by two separate \release{} operations.
	\item \texttt{bool set}\texttt{(int k,} \texttt{data* d):} Sets the version pointed to by $d$ as the current version. Returns \true{} if successful. May also return \false{} if there has been a successful \set[] between this \set[] and the most recent \texttt{acquire(k)}. If the \set[] returns \false{}, it has no effect on the state of the object.
\end{itemize}

We say that a process $p_k$ has \emph{acquired} version $v$ if \acquire$(k)$ returns $v$, and say $p_k$ has \emph{released} $v$ when the next \release$(k)$ operation returns.
If a \set{} operation returns \true, we say that it was \emph{successful}. Otherwise, we say that the \set{} was \emph{unsuccessful} or that the \set{} \emph{aborted}. Note that conditions for correct aborting for the \set{} are reminiscent of $1$-abortability defined by Ben-David \textit{et al.}~\cite{bendavid2016k}, but we relax the requirements to allow a successful \set[] to cause other \set[]s to abort even if it was not directly concurrent with them, but happened sometime since that process's last \acquire[].

An implementation of a \VerMain{} object is considered \emph{correct} if it is linearizable as long as no two operations with the same input $k$ run concurrently.
Furthermore, it is considered \emph{precise} if the \release{} operation returns exactly the versions that stop being live at the moment the \release{} operation returns.
Note that this means that in a precise implementation of the \VerMain{} problem, each \release{} operation $r$ returns a list containing at most one version, and this version must be the one that $r$ released.
We show some properties of a correct \VerMain{} in Appendix \ref{sec:vm-prop}.

Where convenient, for a version $v$, we use \acquire[v], \release[v] and \set[v] to denote an \acquire{} operation that acquires $v$,  a \release{} operation that releases $v$, and a \set{} operation that sets $v$ as the current version, respectively.


\lstset{basicstyle=\footnotesize\ttfamily, tabsize=2, escapeinside={@}{@}}
\lstset{language=C, morekeywords={CAS,commit,empty,local,job,taken,entry,GOTO}}
\lstset{xleftmargin=5.0ex, numbers=left, numberblanklines=false, frame=single}
\makeatletter
\lst@Key{countblanklines}{true}[t]%
{\lstKV@SetIf{#1}\lst@ifcountblanklines}

\lst@AddToHook{OnEmptyLine}{%
	\lst@ifnumberblanklines\else%
	\lst@ifcountblanklines\else%
	\advance\c@lstnumber-\@ne\relax%
	\fi%
	\fi}
\makeatother


\hide{\subsection{Properties of the Version Maintenance Problem}}
\begin{maybeappendix}{prop-proof}
To facilitate presenting the algorithms and their proofs, we begin with a couple observations for version maintenance algorithms.


The sequential specification can be summarized into two points: (1) Each \acquire{} operation returns the data pointer associated with the current version and (2) A \release[v] operation returns true if and only if $v$ is not live after the \release[v] operation. Note that this means that for any particular version $v$, there is exactly one \release{} operation that returns true, and that operation is the \emph{last} operation done on $v$.
In the proofs of linearizability of our algorithm (which we show in Appendix \ref{sec:waitfree-correctness}), we state linearization points, and then proceed to show that for any given history, if we sequentialize it based on the stated linearization points, it adheres to the above sequential specification.

It is useful to note the following two facts, which must hold in any algorithm that solves the version maintenance problem.

\begin{observation}
\label{lem:ss2}
$v$ is live immediately before the linearization point of a \release[v] operation.
\end{observation}

\begin{proof}
Let $C$ be the configuration immediately before the linearization point of a \release[v] operation.
The \release[v] operation must have a corresponding \acquire[v] operation that was linearized before $C$. Since the \release[v] operation is linearized after $C$, $v$ is live at $C$.
\end{proof}

\begin{observation}
\label{lem:ss1}
A version $v$ is live for a contiguous set of configurations.
\end{observation}

\begin{proof}
Note that an \acquire{} operation always returns the current version, which is already alive, and in a \release[v] operation, $v$ is also already alive because it is acquired but not released. Therefore they cannot cause any version to become live. Meanwhile in \set[v] the version $v$ is set to be the current version, thus a version $v$ becomes live only at the linearization point of a \set[v] operation. Since there is only one \set[v] operation in any execution history, $v$ can only become live once, and this completes the proof.
\end{proof}

\hide{
\begin{proof}
To prove this lemma, it suffices to show that $v$ becomes live only at the linearization point of a \set[v] operation. Since there is only one \set[v] operation in any execution history, $v$ can only become live once, and this completes the proof.

Other than a \set[v] operation, the only operations that could potentially cause $v$ to become live are \acquire[v] operations. Since \acquire{} operations always return the current version, they always return a version that is already live. Therefore they cannot cause any version to become live. Therefore $v$ becomes live only at the linearization point of a \set[v] operation.
\end{proof}}
\end{maybeappendix}

\subsection{The \PSWF{} Algorithm}
\label{sec:waitfree}
We now present a simple wait-free algorithm that solves the precise version maintenance problem.
That is, the \release[] operation returns either an empty list of versions, or a singleton containing the version that it is releasing.
We show that our wait-free algorithm is linearizable, and analyze it to obtain strong time complexity bounds; the \acquire[] operation takes $O(1)$ time, and the \release[] and \set[] operations each take $O(P)$ time. Furthermore, we show that in the single-writer setting, where concurrent \set[] operations are disallowed, the algorithm guarantees amortized constant contention per shared-memory operation.
These properties show that regardless of adversarial scheduling, version maintenance need not be a bottleneck for transactions.
The main results are shown in Theorem \ref{thm:waitfree-lin}, \ref{thm:waitfree-step} and \ref{thm:contention}. All proofs are in the Appendix.
Pseudocode for the algorithm is given in Algorithm~\ref{alg:waitfree}, and Figure~\ref{fig:algo} shows how its data is organized.

\begin{figure}
	\includegraphics[width=0.9\columnwidth]{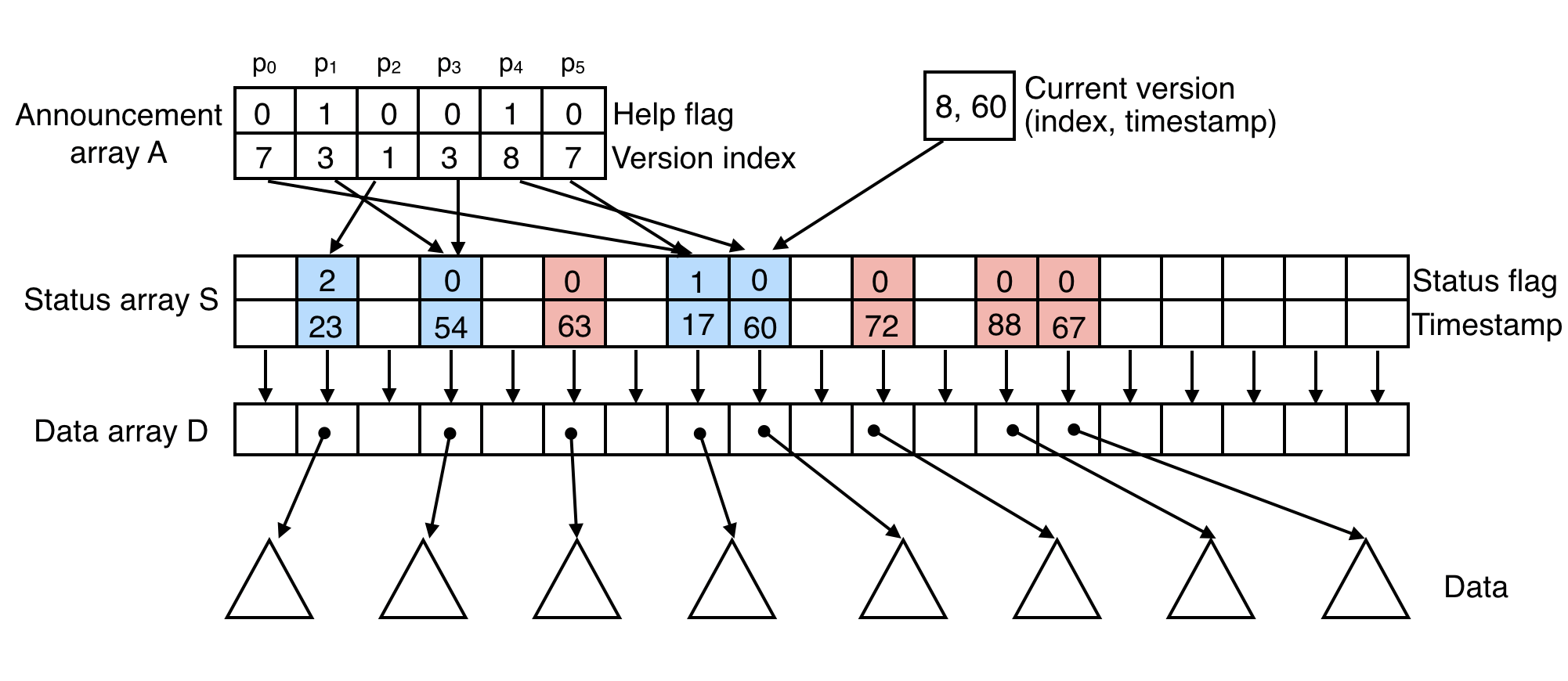}
	\caption{\small The data structures used by Algorithm~\ref{alg:waitfree}. Blue slots in the status array represent live versions. Red slots are versions that a pending \set[] operation is trying to commit. Each announcement array slot has a timestamp in addition to the version index, and each status array slot also has an index, but they are omitted to avoid clutter.}
	\label{fig:algo}
\end{figure}

To understand the idea behind our algorithm, consider the following simplified (but incorrect) implementation. To set a new version, a process $p$ simply CASes its data pointer into a global \emph{currentVersion} location. If its CAS fails then it aborts. To acquire a version, $p$ reads the currentVersion and copies it over to $p$'s slot in an \emph{AnnouncementArray}, thereby signaling to others that it is using this version. The acquire operation then returns the version that it read. When releasing a version $v$, $p$ scans the AnnouncementArray to see whether anyone else is still using $v$. If not, $p$ returns $v$, as it is the last process that used this version. Otherwise, $p$'s release returns an empty list.
This simple outline of an algorithm for the precise \VerMain{} problem satisfies the intuition of what should happen in a solution to the \VerMain{} problem; processes always acquire the current version, and return a version from their release operation only if this version stops being live at the end of the operation. However, this algorithm does not work in a completely asynchronous setting.

To see why, first note that a process $p$ that executes an \acquire{} operation may stall of a long time after reading the currentVersion but before announcing what it read. This could lead to a situation in which, by the time $p$ announces the version $v$ that it read,  $v$ has long since stopped being live, and has already been returned by some \release{} operation. This scenario is not linearizable.
We must also ensure that exactly one releasing process returns each version, meaning that an order between concurrent releasers must be established.
Finally, we need to ensure that if a \set[] aborts, then it or its preceding \acquire[] were concurrent with a successful \set[].

To fix the \acquire{} operation, we assign each process a `helping' flag in its announcement slot, and use that flag to create two stages of the acquire operation; first a version is read from the current version field, \currentv{}, and announced with a `helping' flag set, meaning that this is the version that the process intends to use, but has not started accessing yet. To secure this version, the acquiring process, $p$, must reread the current version to ensure that it has not changed, and then set the `helping' flag to \false. In the meantime, other processes may see $p$'s announcement, and help it complete its acquire.
Some \set[] operations will try to help the \acquire[]s, so that no \acquire[] can repeatedly fail without receiving help.
Once the flag is down, $p$ is said to have \emph{committed} its announced version.
In this way, the releasing process returning the version $v$ can ensure that no process can acquire (commit) the same version $v$ after it terminates.

To ensure that each version is only ever return by one \release{} operation, we assign each version $v$ a ``status'' (stored in the array \texttt{S}), which can be in one of three states at any given time: \emph{usable}, \emph{pending}, and \emph{frozen}.
A \release[v] \op{} mainly deals with two things: helping all other processes complete their \acquire{} on version $v$, when necessary, and deciding if this is the last usage of version $v$, and returning \true{} if so.
If $v$ is \emph{usable}, it means that no \release{} operation is currently in progress on $v$, and $v$ may be in use. If a releasing process $p$ sees this status, it tries to switch its status to pending, and if it succeeds, it then starts scanning the announcement array. While $v$ is \emph{pending}, a single releasing process is scanning the announcement array, and helping any process that has announced $v$ to complete its \acquire. Any releasing process that observes that $v$ is already in the pending state can safely return \false{} because there are currently other processes releasing this version.
Once $p$ has done scanning the array, it sets $v$'s status to \emph{frozen}. This indicates to all other releasing processes that $v$ if no process currently has $p$ acquired, then $v$ can never again be acquired by any new process. Thus, if no process currently has $v$ announced, it is safe to return \true{} on a \release{} of $v$. To ensure that only one releaser does so, the releasers of $v$ compete in erasing $v$ from the status array, and only the winner returns \true{}.

Finally, we allow the \set[] operation by process $p$ to abort only under two conditions: (1) the current version \currentv{} is not the same as $p$'s acquired version (in this case, it is easy to see that there must have been a successful \set[] operation since $p$'s \acquire); or (2) the \set[] operation cannot find a spot in which to place its new version. That is, we have an array called $S$ of versions that are currently active, and it is preallocated with a specific number of slots. Each \set[] operation scans the array of versions to try to find an empty slot in which it can place its new version. The intuition is that if it cannot find an empty slot, then there must have been many other \set[] operations concurrent with it. By setting the size of $S$ to be large enough ($3P+1$ in our case), we can ensure that if a \set[] operation $op$ does not find any empty slots, there must have been some process that has executed a successful \set[] during $op$'s interval.

We now describe the algorithm in more detail. 
A version $v$ is represented as a pair of a timestamp and an index. If $v$ is alive, the status of $v$ is stored in $S[v.index]$ (the \emph{Status} array) and its associated data pointer is stored in $D[v.index]$ (the \emph{VersionData} array).
For the rest of the paper, when we refer to a version, we mean a timestamp-index pair.
Since there are at most $P+1$ live versions, and at most $P$ active \set[] operations that could occupy another slot with a potential version, the Status and Data arrays can never have more than $2P+1$ occupied slots. However, for the purpose of guaranteeing that a \set[] operation will only abort if it was concurrent with a successful \set[], we initialize $S$ and $D$ to be of size $3P+1$.
Each slot $A[k]$ in the announcement array belongs to process $p_k$, and stores a \emph{help flag} \texttt{help} and a version.
A global variable \currentv{} stores the current version.
%
%
%

\myparagraph{Set.} To execute a \set($d$) operation for a data pointer $d$, a process $p$ first creates a new version $v$ locally, and then looks for an empty slot for $v$ in the status array. If it does not find an empty slot, then it aborts. Intuitively, it is ok to abort at that stage because at any given moment, $S$ can have at most $2P$ occupied slots (one version acquired by each process, and another version that is in the middle of being \set[] by each process). So, if $p$ finds all $3P+1$ slots occupied, it means that it was concurrent with $2P+1$ other \set[] operations. Since there are only $P$ processes, at least one process $q$ executed $3$ \set[] operations concurrently with $p$'s \set. If one of $q$'s \set[]s were successful, $p$ can safely abort its own operation. Otherwise, all 3 of $q$'s operations must have been concurrent with a successful \set[] (for $q$ to legally abort), and therefore, at least one of those successful \set[]s must have been concurrent with $p$'s.

Now we assume that $p$ did find an empty slot in $S$. Let $i$ be the index of this empty slot. $p$ initializes $S[i]$ with the new version, and writes $d$ into $D[i]$.
Before setting $v$ as the current version and terminating, $p$ scans the announcement array, and helps every process that needs help
(i.e. $A[k] = \langle true, *\rangle $). To ensure that the helping is successful, $p$ needs to perform three CAS operations on $A[k]$. Each CAS tries to set $A[k]$ to $\langle 0, oldVer \rangle$, where $oldVer$ is the version that $p$ currently has acquired (announced in $A[p]$). To ensure that $oldVer$ is still valid, $p$ checks whether it is still the current version. If it is not, $p$ aborts. These CAS operations can be thwarted at most twice by the \acquire($k$) that requested help, so that the help is guaranteed to have succeeded after the third CAS.
Finally, $p$ tries to set $v$ as the current version by CASing it into \currentv{}. If this CAS succeeds, so does $p$'s \set[] operation. If it fails, $p$ aborts, but first clears the slot it occupied in $S$ to allow others to use it.

\myparagraph{Acquire.} 
The \acquire($k$) operation begins by requesting help, reading the current version $v$, and announcing it in $A[k]$. To ensure that $v$ is still the current version at the announcing step, the operation reads \currentv{} again. There are two cases. If it finds that \currentv{} has been updated, it starts over. It will only ever restart once, because
if it finds that \currentv{} has been updated once again, it knows that two \set{} operations have occurred, one of which must have committed a version into $A[k]$ by performing 3 helping CASes.
If $v$ is still the current version, we use a CAS to set the helping flag in $A[k]$ to $0$. Even if this CAS fails, $A[k]$'s helping flag must now be $0$, since an \acquire's CAS only fails if it was helped by another process (a \set{} or a \release{} operation).
Once \acquire($k$) successfully commits a version $v$, 
it reads and returns the corresponding data pointer $D[v.index]$.


\myparagraph{Release.} To perform a \release($k$) operation, the process $p_k$ first reads the committed version $v$ from its announcement slot, and clears the slot. If $v$ is still current, the \release($k$) operation returns \false{} because $v$ is still live. Otherwise,
it must check whether someone else is still using $v$.
This is done by looking at the status at $S[v.index]$.
$S[v.index]$ might be empty or store a version other than $v$. In that case, some other \release{} of $v$ has already returned \true{}, so $p_k$ returns \false{}.
Otherwise, if $S[v.index]$ stores a valid status (\emph{usable}, \emph{pending}, or \emph{frozen}), then $p_k$ uses this status to determine what to do, as described earlier.

\renewcommand{\figurename}{Algorithm}
\lstset{literate={<}{{$\langle$}}1  {>}{{$\rangle$}}1}

\begin{figure*}[!t!h]
	\caption{The Precise, Safe and Wait-free Algorithm for the Version Maintenance Problem}
\begin{minipage}{0.48\textwidth}
\begin{lstlisting}[countblanklines=false]
enum VStatus {usable, pending, frozen};
struct Version{
	int timestamp;
	int index; 	};
struct VersionStatus {
	Version v;
	VStatus status; };
struct Announcement {
	Version v; 	
	bool help;	};

Version V;
VersionStatus S[3P+1];
Announcement A[P];
Data* D[3P+1];
Version empty = <@$\bot$@, @$\bot$@>;
Data* getData(Version v) {
	return D[v.index];} @\label{line:wf_read}@
	
bool set(int k, Data* data) {
    //the version you acquired	
    Version oldVer = A[k].v; 
	Version newVer;
    //find empty slot
	for(int i = 0; i @$<$@ 3P+1; i++) { 
		if(S[i] == <empty, usable>) {
			newVer = <V.timestamp+1, i>;
			if(CAS(S[i],<empty, usable>, @\label{setFillS}@
                    <newVer, usable>)){ 
				D[i]=data;
				break;	}   }	
		if(i == 3P) return false; } @\label{setAbort}@
    //try to help everyone
	for(int i = 0; i @<@ P; i++) { 
        //help 3 times
		for(int j = 0; j @<@ 3; j++) { 
			Announcement a = A[i];
			if(a.help) {
				if(oldVer != V) return false; @\label{setHelpCheck}@
				CAS(A[i],a,<oldVer, false>);@\label{set_help}@ } } }
	bool result = CAS(V, oldVer, newVer); @\label{lin_set}@
	if (!result){
		S[i] = <empty, usable>;	}
	return result;	}
\end{lstlisting}
\end{minipage}\hfill
\begin{minipage}{0.49\textwidth}
\begin{lstlisting}[countblanklines=false,firstnumber=37]
Data* acquire(int k) {
	Version u = V; //read current version V
	A[k] = <u, true>; //announce it @\label{request_help}@
	if(u == V) {
		CAS(A[k], <u, true>, <u, false>);
		return getData(A[k].v); }
    //try again with new version
	for(int i=0;i@<@2;i++){ 
		Version v = V;
		if(!CAS(A[k], <u, true>, <v, true>)) { @\label{line:wf_r1}@
			return getData(A[k].v); }
		if(v == V) {
			CAS(A[k], <v, true>, <v, false>); @\label{line:wf_r2}@
			return getData(A[k].v); }
		u = v; }
	return getData(A[k].v); 
}

data** release(int k) {
	Version v = A[k].v;
	A[k] = <empty, false>;  @\label{clear}@
	if(v == V) return null;  @\label{still_current}@
	VersionStatus s = S[v.index];	@\label{read_d}@
	if (s.v != v) return null  @\label{already_collected}@
	if (s.status == usable) {
		if(!CAS(S[v.index],s,<s.v,pending>)){
			 return null;}  @\label{compete_to_help}@
		for(int i = 0; i @<@ P; i++) @\label{help_loop}@ {
			Announcement a = A[i];
			if(a == <v, true>) {
				CAS(A[i], a, <v, false>);} @\label{relhelp}@ }
		s = <s.v, frozen>;
		S[v.index] = s;	} @\label{h_two}@
	if (s.status == frozen) {  @\label{check_h_two}@
		for(int i = 0; i @<@ P; i++) 
			if(A[i] == <v, false>) {
				return null;}  @\label{relcheck}@
		if (CAS(S[v.index],s,<empty,usable>)) { @\label{lin_rel}@
			return [v];	}	
		else return null;	}
	return null; @\label{h_one}@ 
}
\end{lstlisting}
\end{minipage}
\label{alg:waitfree}
\end{figure*}

\renewcommand{\figurename}{Figure}




\iffull
In Section \ref{sec:waitfree-proof} of the Appendix,
we show that our version maintenance algorithm is correct and efficient; we show that the algorithm is linearizable, and the amortized time per operation is low, even when accounting for contention.
In particular, we show the following two theorems.
\subsection{Correctness and Bounds of the Wait-free Algorithm}
In this section,
we show that our version maintenance algorithm is correct (linearizable) and efficient.
\fi
This algorithm can be shown to be correct (linearizable) and efficient. We summarize the results as follows:

\begin{theorem}[Correctness]
	\label{thm:waitfree-lin}
	Algorithm \ref{alg:waitfree} is a linearizable solution to the Version Maintenance Problem.
\end{theorem}

\begin{theorem}[Step bounds]
	\label{thm:waitfree-step}
	Each \acquire() operation requires at most $O(1)$ time and each \release() and \set() operation requires $O(P)$ time.
\end{theorem}

\begin{theorem}[Amortized Contention]
	\label{thm:contention}
	When concurrent \set[] operations are disallowed, each \acquire() operation experiences $O(1)$ amortized contention and each \release() and \set() operation experiences $O(P)$ amortized contention. Furthermore, no contention experienced by \acquire() is amortized to \release() or \set().
\end{theorem}

We show proofs in Appendix \ref{sec:waitfree-proof}.
We note that Theorem~\ref{thm:contention} shows a property that is non-trivial to be achieved in wait-free algorithms, even in the single-writer setting---regardless of the adversarial scheduler, processes do not often contend on the same operations.
Intuitively, our algorithm achieves this because of the version status: instead of allowing many releasing processes to traverse and modify the announcement array for every version, only one process per version (the one that changed the status from \emph{usable} to \emph{pending}) can do this at any given time. Furthermore, each slot in the announcement array can only have one version associated with it at any given time, meaning that only one releaser, one acquirer, and one setter can contend on any given slot.

\begin{maybeappendix}{waitfree-correctness}
\myparagraph{Correctness.} Our wait-free algorithm is linearizable. We first state the linearization points of the \op{}s, which can be used to sequentialize any history and lead to an execution that is consistent with the sequential specification of the \VerMain{} problem.
	We begin the correctness proof with a few observations and lemmas that will help us understand the general flow of the algorithm.
	After presenting these invariants, we state the linearization points of the \op{}s. We then proceed to show that, if we sequentialize any given history according to our stated linearization points, then we obtain an execution that is consistent with the sequential specification of the version maintenance problem.
	
	First, 	note that the timestamp mechanism works ``as expected'': The timestamp in \currentv{} is always increasing. Furthermore, this is true for every slot $i$ of the version status and announcement arrays in isolation. That is, for all $i$, the timestamps in $S[i].v$ and $A[i].v$ are monotonically increasing.
	
	We say that a \set[] operation is \emph{successful} if it returns \true. Otherwise we say it was \emph{unsuccessful}. We now show a simple but useful lemma about the \set[] method.
	
	\begin{lemma}\label{lem:no-concur-set}
		No two successful \set[] operations are concurrent.
	\end{lemma}

	\begin{proof}
		Let $S$ be a successful \set[] operation executed by process $p$. Assume by contradiction that there is some successful \set[] operation $S'$ concurrent with $S$. One of these \set[] operations must have executed line~\ref{lin_set} first. Assume without loss of generality that  Note that $s'$ executed line~\ref{lin_set} first.
		
		For a \set[] operation to be successful, it must execute a successful CAS in line~\ref{lin_set}. That is, at the time of its CAS in line~\ref{lin_set}, the current version \currentv{} is equal to \texttt{oldVer}. Note that \texttt{oldVer} is the version that $p$ acquired in its last \acquire[] operation. Therefore, for $S$ to be successful, \currentv{} must not have changed since $p$ acquired its version before $S$ was invoked. Therefore, $S'$ must have executed line~\ref{lin_set} before $S$ was invoked.
		Note that every successful \set[] operation changes the current version in line~\ref{lin_set}, and terminates immediately afterwards without making any more changes to shared memory. So, $S'$ was not concurrent with $S$, leading to a contradiction.
	\end{proof}
		
	Note also that for any slot $k$ in the announcement array, $A[k] = \langle 1, *\rangle$ only if process $k$ is currently executing an \acquire[]($k$) \op. This is easy to see, since the help flag never gets set to $1$ from any other functions, and a process $k$ only ever accesses its own slot when executing \acquire[].
	We now show that every completed call to \acquire[] commits exactly one version.
	
	\begin{lemma}
		\label{commit}
		Let $Q$ be an \acquire(k) \op{}. There exists exactly one successful CAS (can be called by either \acquire{}, \set{} or \release{}) during $Q$ that sets $A[k].help$ to $0$. This CAS occurs before $Q$ calls \getdata().
	\end{lemma}
	
	\begin{proof}
		We first prove that such a CAS exists, then we prove that it is unique. If a CAS in $Q$ fails, then we are done because that CAS must have been interrupted by a successful helping CAS from a \release() or a \set() operation which sets $A[k].help$ to 0. Now suppose all CAS \op{}s in $Q$ succeed. If $Q$ performs a CAS that sets $A[k].help$ to 0, then we are done.
		Otherwise, $Q$ must have read \currentv{} twice and found that the version it read was out of date both times. This means that it must have been interrupted by line \ref{lin_set} of two different \set() operations (since \set[]s are the only way that versions can change). By Lemma~\ref{lem:no-concur-set}, between the two executions of line \ref{lin_set}, a \set() method tries to help operation $Q$ up to 3 times by performing a CAS onto $A[k]$. At most 2 of these helping CAS \op{}s can fail due to a CAS from $Q$. So assuming that there is no interference from \release() operations, one of the 3 helping CAS \op{}s is guaranteed to succeed and set $A[k]$ to 0. If there is a successful CAS onto $A[k]$ by a \release() operation in this interval, then $A[k].help$ will be set to 0 by CAS \op{} from \release() instead. Therefore, in every possible case, there exists a successful CAS \op{} $O$ during $Q$ that sets $A[k].help$ to 0. Note that this CAS always happens before $Q$ calls \getdata().

		Since $A[k].help$ changes from 0 to 1 only on line \ref{request_help} of \acquire(k), if $A[k].help$ is set to 0 during $Q$, then it remains unset until the end of $Q$. While $A[k].help = 0$, there cannot be a successful CAS \op{} on $A[k]$. Therefore there exists a unique CAS during $Q$ that sets $A[k].help$ to $0$.
	\end{proof}

		
	We refer to the CAS that commits a version for an \acquire[] \op{} $Q$ the \emph{committing CAS} for $Q$. Now that we know this fact about the \acquire[] \op{s}, we can state the linearization points of each \op.
		
	\begin{definition}\label{def:linpoint}
		For each \op{} $op$, its \emph{linearization point} is as follows:
		\begin{itemize}
			\item If $op$ is a successful \set[], then it is linearized at line \ref{lin_set}. This is the line that updates \currentv{}.
			\item If $op$ is an unsuccessful \set[], it linearizes at its return.
			\item If $op$ is an \acquire[], there are two cases:
			(1) If the committing CAS is from a \set[], then $op$ is linearized at line~\ref{setHelpCheck} of the helping \set[].
			(2) Otherwise, let $v$ be the version committed by the committing CAS of $op$. $op$ is linearized at the instruction that it reads $v$ from \currentv{} (either on line \ref{line:wf_r1} or \ref{line:wf_r2}).
			\item If $op$ is a \release{}, then there are two cases:
			(1) If $op$ completes and returns \true, then it is linearized at its final instruction.
			(2) Otherwise, it is linearized when it performs line \ref{clear}.
		\end{itemize}
	\end{definition}
%

	We now prove that given any history $H$, sequentializing $H$ according to the linearization points yields a history that follows the sequential specification of the version maintenance problem. The sequential specification is outlined in Appendix \ref{sec:vm-prop}.
	The first step is to show that \acquire[] \op{s} behave as specified. That is, if both \set[]s and \acquire[]s are linearized as in Definition \ref{def:linpoint}, then each \acquire[] always returns data pointer associated with the current version. To prove this, we first prove the following 2 lemmas which help us argue that $D[v.index]$ stores the data pointer associated with $v$ as long as $v$ is committed for some process. In the proof of both lemmas, we repeatedly use the fact that for a fixed value of $S[i].v$, $S[i].h$ is strictly increasing (it starts at $0$ and goes up to $2$).

	\begin{lemma}
		\label{waitfree-h4}
		If $\langle 0, v\rangle$ is written to some element $A[k]$ of $A$ at time $t$, then one of the following conditions hold:
		\begin{enumerate}
			\item $S[v.index] = \langle v, 0\rangle$ at time $t$, or
			\item $S[v.index] = \langle v, 1\rangle$ at time $t$, or
			\item there is some other element $A[j]$ such that $j\neq k$ and $A[j] = \langle 0, v\rangle$ at time $t$.
		\end{enumerate}
	\end{lemma}

	\begin{proof}
		Let $k$ be any process id. We first show that $\langle 0, v\rangle$ is never written to $A[k]$ after $S[v.index]$ is set to $\langle v, 2\rangle$ unless there is already an element of the array $A$ whose value is $\langle 0, v\rangle$. Let $C$ be the earliest configuration in which $S[v.index] = \langle v, 2\rangle$. From the code, we can see that $C$ occurs immediately after line \ref{h_two} of some \release[v]() operation $R$. We first prove that no \release[v]() operation tries to set $A[k] = \langle 0, v\rangle $ after this configuration. Looking at the code, we see that $R$ does not write to $A[k]$ after configuration $C$. We also know that $R$ must have succeeded in the CAS on line \ref{compete_to_help}. We claim that no \release[v]() operation other than $R$ will execute the loop on line \ref{help_loop}. This is because at most one process can succeed in the CAS on line \ref{compete_to_help}, since $S[v.index]$ is never changed back to $\langle v, 0\rangle $ after it has been set to $\langle v, 1\rangle $. By the claim, we know that no \release[v]() operation tries to set $A[k]$ to $\langle 0, v\rangle $ after $C$.
		
		Note that a \set[] operation by process $p_j$ only helps \acquire[] operations by trying to commit $p_j$'s own acquired value \texttt{oldVer}. At the time $p_j$ is executing the \set[] operation, it has already finished acquiring, but not yet started releasing \texttt{oldVer}. Therefore, $A[j] =\langle 0, v\rangle$ at the time $p_j$ helps $p_k$'s acquire.
		
		Finally, suppose for contradiction that some \acquire[v](k) operation $Q$ changes $A[k]$ to $\langle 0, v\rangle $ after $C$. $Q$ first sets $A[k]$ to $\langle 1, v\rangle $. If this happens before $R$ executes the $k$th iteration of line \ref{relhelp}, then $A[k]$ equals $\langle 0, v\rangle $ before $C$, so it cannot change to $\langle 0, v\rangle $ after $C$. If this happens after $R$ executes the $k$th iteration of line \ref{relhelp}, then after $Q$ sets $A[k]$ to $\langle 1, v\rangle $, it will see that $v$ is no longer the current version and not attempt to set $A[k]$ to $\langle 0, v\rangle $, another contradiction. Therefore none of the operations will change the value of $A[k]$ to $\langle 0, v\rangle $ after configuration $C$.

		Now all that is left is to show that $\langle 0, v\rangle$ cannot be written to $A[k]$ before $S[v.index] = \langle v, 0\rangle$. Let $C$ be some configuration before $S[v.index] = \langle v, 0\rangle$. At this configuration, the version $v$ has not yet been written to $V$, so there's not way for it to be written to $A[k].v$.
	\end{proof}

	\begin{lemma}
		\label{waitfree-h1}
		If $A[k] = \langle 0, v\rangle$ for some $k$, $D[v.index]$ stores the data pointer associated with $v$.
	\end{lemma}

	\begin{proof}
		Let $d$ be the data pointer associated with $v$. By Lemma \ref{waitfree-h4}, $S[v.index] = \langle v, 0\rangle$ or $S[v.index] = \langle v, 1\rangle$ when $\langle 0, v\rangle$ is written to $A[k]$. From the code, we can see that $D[v.index]$ is not written to unless $S[v.index]$ is empty. Therefore, $D[v.index] = d$ at the step that writes $\langle 0, v\rangle$ to $A[k]$. In order for a \release() operation to empty $S[v.index]$, it must pass the checks on line \ref{relcheck}. These checks can only occur after $S[v.index]$ is set to $\langle v, 2\rangle$, and only pass if no index of $A$ contains $\langle v, 0\rangle$.
		Therefore as long as $A[k] = \langle 0, v\rangle$, the checks on line \ref{relcheck} fail and $S[v.index]$ will not be emptied. As a result, $D[v.index] = d$ as long as $A[k] = \langle 0, v\rangle$.
	\end{proof}

		
	Next, we prove the first part of the sequential specification, which says that \acquire[] \op{s} return the correct value.

	\begin{lemma}
		\label{q1}
		An \acquire(k) \op{} $Q$ returns the pointer written by the last \set() \op{} linearized before it.
	\end{lemma}
	
	\begin{proof}
		Note that at any point in an execution, the global variable \currentv{} stores the \vptr{} created by the last \set() operation linearized before this point. Furthermore, at any point in the execution, $D[\currentv{}.index]$ stores the pointer written by the last \set() \op{} linearized before this point. This is because $D[i]$ cannot change as long as $S[v.index]$ is not empty and $S[v.index].h$ cannot be empty as long as $i = \currentv{}.index$ (due to the check on line \ref{still_current}).

		Let $v_0$ be the value of \currentv{} at $Q$'s linearization point. It's easy to see that $Q$ commits $v_0$ by looking at how \acquire{} operations are linearized. By Lemma \ref{commit}, $Q$ calls \getdata() with the argument $v_0$. Since $A[k] = \langle 0, v_0\rangle$ during the call to \getdata(), by Lemma \ref{waitfree-h1}, we know that $Q$ returns the data pointer associated with $v_0$.

	\end{proof}



	We prove part (2) of the sequential specification in two parts: if a \release[v] \op{} returns \false{}, then $v$ is still live after that \op{}, and if the \release[v] returned \true{}, then $v$ stops being live immediately afterwards. Since every \release[v] \op{} must start when $v$ is live, this means that once it loses that status, no other \release[] \op{} can be executed on it. Therefore, the \release[v] that returns \true{} is the last one.
	To be able to formally show these properties, we begin with considering which configurations in the execution imply that $v$ is live.
	This lemma relies on the definition of a \emph{live} version (Definition \ref{def:live}).
	
	\begin{lemma}
		\label{waitfree-h5}
		A version $v$ is live at a configuration $C$ if at least one of the following holds:
		\begin{enumerate}
			\item It is the committed version for some process. That is, $A[k] = \langle 0, v\rangle $ for some $k$.
			\item Its timestamp is smaller than that of the current version, it is written in the version status array, and its collecting flag is not yet set to $2$.
			That is, \currentv$.timestamp \geq v.timestamp$, $S[v.index].v = v$, and $S[v.index].h < 2$.
		\end{enumerate}
	\end{lemma}
	
	\begin{proof}
		We consider each condition separately.
		\begin{enumerate}
			\item Assume the first condition holds. From the code, we can see that $A[k] = \langle 0, v\rangle $ at configuration $C$ only if an \acquire[v](k) operation has been linearized but line \ref{clear} of the corresponding \release[v](k) operation has not occurred yet. Since \release() operations are linearized at line \ref{clear} or later, the corresponding \release[v](k) operation must be linearized after $C$. Therefore $v$ is live at $C$ (condition (2) of the definition of \emph{live}).
			
			\item Assume the second condition holds. Let $S$ be the \set[v]() operation that made $v$ the current version. The configuration $C$ must occur after $S$ is linearized because we know that $V.timestamp \geq v.timestamp$ at $C$. Suppose for contradiction that $v$ is not live at $C$. This means there exists a \set() operation $S'$ linearized before $C$ that changed the current version from $v$ to something else. This also means that each \acquire[v]() operation linearized before $C$ has a corresponding \release[v]() operation linearized before $C$. Since no \acquire[v]() operation is linearized after $C$, no \release[v]() operation is linearized after $C$. If some \release[v]() operation returns 1, then $S[v.index] = \langle v, 2\rangle $ immediately before the linearization point of that operation. Since $C$ occurs after this linearization point, this would contradict the fact that $S[v.index].v = v$ and $S[v.index].h <  2$ at $C$. Thus to complete the proof, it suffices to show that some \release[v]() operation returns 1.
			
			Since we always release the previous version after setting the new version, there is a \release[v]() operation $R'$ that starts after $S'$ completes and that is linearized before $C$. There are two possibilities: either this \release[v]() operation $R'$ sets $S[v.index]$ to $\langle v, 1\rangle $ or it has already been done by some other \release[v]() operation. Let $R_1$ be the \release[v]() operation that sets $S[v.index]$ to $\langle v, 1\rangle $. $R_1$ sets $S[v.index]$ to $\langle v, 2\rangle $ on line \ref{h_two}, so it enters the if statement on line \ref{check_h_two}. Since $S[v.index].v = v$ and $S[v.index].h < 2$ at $C$, we know that $C$ occurs before $R_1$ executes line \ref{h_two}. If $R_1$ returns on line \ref{relcheck}, then it sees that $A[k] = \langle 0, v\rangle $, so $v$ is live at the end of $R_1$ by part (1). Therefore $C$ occurs between two live configurations, so by Observation \ref{lem:ss1}, $v$ is also live at $C$, a contradiction. Finally, if $R_1$ returns on line \ref{lin_rel}, then either $R'$ returns 1 or its CAS is interrupted by a \release[v]() operation that returns 1 and this completes the proof.	\qedhere
		\end{enumerate}		
	\end{proof}

	Now we are ready to prove part (2) of the sequential specifications.
	
	\begin{lemma}
		\label{r1}
		If a \release[v](k) operation $R$ returns \true{} then $v$ is live before the linearization point of $R$ and not live after.
	\end{lemma}
	
	\begin{proof}
		By Observation \ref{lem:ss2}, $v$ must be live before the linearization point of $R$. 

		Let $R$ be a \release[v](k) operation that returns \true{}. $R$ must be linearized at line \ref{lin_rel}, because $R$ is linearized at its last instruction and this is the only line in $R$ that could potentially return \true{}. In order to reach this line, the check on line \ref{still_current} must return \false, so we know that $v$ is not the current version at the linearization point of $R$. Let $C$ be the configuration immediately after the linearization point of $R$. To show that $v$ is not live at $C$, we just need to show that each \acquire[v]() operation linearized before $C$ has a corresponding \release[v]() operation that was also linearized before $C$.

		We first claim that there is at most one \release[v]() operation that returns \true. To see why this is true, recall from algorithm \ref{alg:waitfree} that a \release[v]() operation returns \true{} if and only if it successfully changes $S[v.index]$ from $\langle v, 2\rangle $ to empty on line \ref{lin_rel}. Once $S[v.index]$ is emptied, it will never again store the value $\langle v, 2\rangle $. Therefore at most one \release[v]() operation can succeed in the CAS on line \ref{lin_rel}, so at most of \release[v]() operation can return \true{}.
		
		Since $R$ succeeds in changing $S[v.index]$ from $\langle v, 2\rangle$ to empty, we know that $S[v.index] = \langle v, 2\rangle$ at line \ref{check_h_two} of $R$. By Lemma \ref{waitfree-h4}, no \acquire() operation can commit the version $v$ after this line. Therefore all \acquire[v]() operations linearized before $C$ have already committed the version $v$ before line \ref{check_h_two} of $R$. In order for $R$ to reach line \ref{lin_rel}, it must pass all of the checks on line \ref{relcheck}. This means that $\langle 0, v\rangle$ must have been cleared from $A[k]$ before the $k$th iteration of line \ref{relcheck} by $R$. Once an \acquire[v]($k$) operation sets $A[k] = \langle 0, v\rangle$, it can only be cleared by line \ref{clear} of the corresponding \release[v]($k$) operation. By the earlier claim, all \release[v]() operations other than $R$ are linearized on line \ref{clear}. Therefore each \acquire[v]() operation linearized before $C$ has a corresponding \release[v]() operation that was also linearized before $C$.

		Therefore $v$ is not live after $C$ and the lemma holds.
	\end{proof}
	
	\begin{lemma}
		\label{r2}
		If a \release[v](k) operation $R$ returns \false{} then $v$ is live after the linearization point of $R$.
	\end{lemma}
	
	\begin{proof}
		Recall that $R$ is linearized at line \ref{clear}. By Observation \ref{lem:ss2}, we know that $v$ is live just before line \ref{clear} of $R$. Therefore to show that $v$ is live immediately after this line, it suffices to show that $v$ is live at some configuration after this line (By Observation \ref{lem:ss1}).
		
		Suppose $R$ returns on line \ref{still_current}. Then $v$ is still the current version at line \ref{still_current} of $R$, which means it is live at line \ref{still_current}.
		
		Next we prove the following claim. If $R$ sees that a \release[v]() operation $R'$ returning 1 has already been linearized, then $v$ is live after the linearization point of $R$. By lemma \ref{r1}, $v$ is live before the linearization point of $R'$ and not live after. Since $v$ is live before the linearization point of $R$, by Observation \ref{lem:ss1}, we know that $R'$ is linearized after $R$, so $v$ is live after the linearization point of $R$.
		
		Suppose $R$ returns on line \ref{already_collected}. If the check on that line returns \true, then $S[v.index]$ is either empty or it has already been reused by a \set() operation for a newer version. In both cases, some \release[v] operation $R'$ has succeeded in its final CAS and returned 1. By the previous claim, $v$ is live after the linearization point of $R$.
		
		Suppose $R$ returns on line \ref{lin_rel}. Then the CAS \op{} on line $\ref{lin_rel}$ must have failed for $R$. That means that some \release[v] operation $R'$ has already succeeded in performing this CAS \op{}, so by the previous claim, $v$ is live after the linearization point of $R$.
		
		
		Suppose $R$ returns on line \ref{compete_to_help} or line \ref{h_one}. Then $S[v.index].v = v$ and $S[v.index].h < 2$ on line \ref{read_d} of $R$. Furthermore, $v$ is no longer the current version, which means that $V.timestap > v.timestamp$. So by lemma \ref{waitfree-h5}, $v$ is live when $R$ executes line \ref{read_d}.
		
		Suppose $R$ returns at the $i$th iteration of line \ref{relcheck}. Since $R$ sees that $A[k] = \langle 0, v\rangle $ on this line, by lemma \ref{waitfree-h5}, $v$ is live when $R$ executes this line.	
	\end{proof}
	
	The last remaining part of the proof is to show that unsuccessful \set[] operations are correct; that is, we show that if \set[] operation by process $p$ returned \false{}, then either $p$'s last \acquire[] or this \set[] must have been concurrent with a successful \set[] operation. Note that this property is very similar to $1$-abortability~\cite{bendavid2016k}. In fact, our condition is stronger than $1$-abortability, since we only allow another \set[] operation to prevent a \set[] operation from succeeding.
	
	\begin{lemma}\label{lem:setAbort}
		If a \set[] operation $Q$ by process $p$ is unsuccessful, then $p$'s \acquire[]-\set[] pair must be concurrent with some successful \set[] operation.
	\end{lemma}

	\begin{proof}
		Let $Q$ be an unsuccessful \set[] operation by process $p$.
		Note that there are 3 places in which an unsuccessful \set[] may return: line~\ref{setAbort}, line~\ref{setHelpCheck}, and line~\ref{lin_set}. We consider each of these separately.
		\begin{itemize}
			\item If $Q$ returned the result of the CAS on line~\ref{lin_set} and this value was \false, there must have been a successful CAS on \currentv{} since $p$ obtained the expected value \texttt{OldVer}. That is, there must have been a successful \set[] that changed \currentv{} since $p$'s last \acquire[].
			\item If $Q$ returned on line ~\ref{setHelpCheck}, there must have been a successful \set[] that changed \currentv{} since $p$'s last \acquire[], exactly the same as the previous case.
			\item If $Q$ returned at line ~\ref{setAbort}, then it did not find an empty slot in the array $S$. Note that at any point in time, $S$ can only have at most $2P$ slots full; one slot for the acquired version of each process, and one slot for the new version that each process is trying to commit in its ongoing \set[] operation. Thus, during $Q$'s scan of the array $S$, earlier slots must have cleared and new slots must have been filled. Note that a new slot of the array $S$ can only be filled by a \set[] operation, in line~\ref{setFillS}. Furthermore, at most one slot is filled by each \set[] operation, and an unsuccessful \set[] operation clears its own slot before terminating. So, for $Q$ to observe all $3P+1$ slots full, there must have been at least $P+1$ new slots filled during its execution, meaning that there were at least $P+1$ \set[] operations that started after $Q$ started. Since there are only $P$ processes, there must have been at least one process that started two new \set[] operations during $Q$'s scan of the array. Furthermore, this process also had an ongoing \set[] operation when $Q$ started. Let this process be $p_j$. Consider $p_j$'s middle \set[] operation; that is, the \set[] by $p_j$ whose interval is completely contained within $Q$'s interval. Call this \set[] operation $Q'$. If $Q'$ succeeded, then we are done. Otherwise, $Q'$ must have aborted, but not on line~\ref{setAbort} (since we know that $Q'$ filled a slot of $S$). So, by the other two cases, either $Q'$ or $p_j$'s previous \acquire[] operation was concurrent with a successful \set[] operation. Since we know that $Q$ was concurrent with $p_j$'s previous \set[] operation, it must have also been concurrent with $p_j$'s last \acquire[] (by the correct order of operations on a VM object). Thus, $Q$ was concurrent with the successful \set[] that made $Q'$ abort. \qedhere
		\end{itemize}
	\end{proof}
	
	Together, Lemmas \ref{q1}, \ref{r1}, \ref{r2} and \ref{lem:setAbort}, and Definition \ref{def:linpoint} directly imply
	Theorem \ref{thm:waitfree-lin}.
\end{maybeappendix}

\begin{maybeappendix}{waitfree-time}
	\myparagraph{Time Complexity.}
	From a quick inspection of the code, it is easy to see that \acquire{} takes $O(1)$ machine operations, while both \release{} and \set{} each take $O(P)$ machine operations, where $P$ is the number of processes in the system. This proves Theorem \ref{thm:waitfree-step}.

	\myparagraph{Amortized Contention.}
	More interestingly, we now show that the machine \op{}s are not heavily contended in the single-writer case where only the writer can perform \set{}s.
	In particular, we show that each operation in our algorithm experiences low contention on average.
	Recall from Section \ref{sec:prelim} that the amount of contention experienced by a single \op{} $i$ in a history $H$ is defined as the number of \emph{responses} to modifying \op{}s on the same location that occur between $i$'s invocation and response in $H$.
	
	In order to easily discuss accesses to the Announcement array, we define a \release{} \op{}'s \emph{helping CAS \op{}s} as the CASes that it does in line \ref{relhelp}. Note that some \release{} \op{}s do not execute any helping CASes at all.
	We begin with a simple observation.
	
	\begin{observation}\label{obs:quick_release}
		For each version $v$ that is announced in the Announcement array, there is at most one \release[v] \op{} that executes helping CASes for that version.
	\end{observation}
	
	This is due to the fact that releasing processes of the same version $v$ compete to change $v.h$ to $1$ and only the process whose CAS succeeds goes on to execute any helping CASes.
	Furthermore, note that a process executing \acquire[] executes at most $3$ CAS \op s before returning. Two of these CAS \op{}s may announce a different version in the array, and the last CAS may commit the version.
	This leads to the following observation.
	
	\begin{observation}\label{obs:quick_acq}
		An \acquire[] \op{} announces at most $2$ versions in the Announcement array.
	\end{observation}
	
	To show that our algorithm has low contention, the hardest part is showing that there is not too much contention on the announcement array $A$. At first glance, it looks like there might be a bad execution where half of the processes are running the helping portion of the \release() method and they contend at each elements of $A$. However in the next lemma, we take a step towards showing that this is not possible by proving that the number of CAS instructions on $A$ is at most 8 times the number of \acquire() operations.


	\begin{lemma}\label{lem:cas_ub}
		If no two \set{}s are concurrent and $a$ is the number of \acquire($k$) operations by process $p_k$, then $A[k]$ experiences at most $8a$ CAS instructions.
	\end{lemma}

	\begin{proof}
		As discussed above, each \acquire[](k) \op{} can perform up to $3$ CAS \op{}s, all of them on $A[k]$. Thus, in total, there can be at most $3a$ CAS \op{s} on $A[k]$ from \acquire{} \op{s}.

		Next we show that there are at most $3a$ CAS instructions on $A[k]$ from \set($k$) operations. A \set() operation tries to help process $p_k$ only if it sees that the flag $A[k].help$ is set. This flag must have been set by some \acquire($k$) operation and it will either be unset by the helping \set() operation or it will be unset by some other operation during the helping \set() operation. In either case, this \acquire($k$) operation receives help before the \set() returns so it will not be helped by any future \set() operation. Since \set{}s cannot be concurrent, each \acquire($k$) has at most one helping \set() operation. Each helping \set() operation performs at most $3$ CAS instructions on $A[k]$, so there are at most $3a$ CAS instructions on $A[k]$ from \set($k$) operations.

		Now we just need to show that there are at most $2a$ CAS instructions on $A[k]$ from \release() operations. A \release[v]() operation helps process $p_k$ only if it sees that $p_k$ has announced the version $v$. By Observation \ref{obs:quick_acq}, at most $2a$ different versions are announced to $A[k]$ and by Observation \ref{obs:quick_release}, for each version that is announced, at most one \release() operation tries to help process $p_k$. Therefore process $p_k$ is helped by at most $2a$ \release() operations. Each helping \release() operation performs a single CAS on $A[k]$, so there are at most $2a$ CAS instructions on $A[k]$ from \release() operations.
	\end{proof}

	Now that we have shown that each \acquire() operation leads to a constant number of CAS instructions on $A$, we can use the fact that \acquire() operations are always followed by \release() operations to argue that there are (approximately) a constant number of CAS instructions on $A$ for each \release() operation as well. Since each CAS operation causes at most $P$ contention, we can argue that the total amount of contention on $A$ is at most $O(P)$ times the number of \release() operations. The proof of Theorem \ref{thm:contention} formalizes this argument and fills in the other details.

	\begin{proof}[Theorem \ref{thm:contention} Proof]
		Let $N$, $M$, and $L$ be the number of \acquire(), \release() and \set() operations, respectively. It suffices to show that the amount of contention experienced by all operations is $O(N+MP+LP)$ and that the amount of contention experienced by \acquire() operations is $O(N)$. We consider the amount of contention on the four global variables, $V$, $S$, $D$ and $A$, separately. Only \set() operations write to variables $V$ and $D$, and there can only be one \set() operation at a time, so each access to $V$ or $D$ experience constant contention. Each operation accesses $V$ and $D$ a constant number of times, so these two variables contributes $O(N+M+L)$ to the total contention and $O(N)$ to the contention experienced by \acquire() operations.

		Next, we consider the amount of contention on the array $S$. Each \release() and \set() performs a constant number of writes to the array $S$. Each write causes at most $P$ contention, so these writes add at most $O(MP + LP)$ to the overall contention. Note that read operations do not cause any contention. \acquire() operations never access $S$, so they experience no contention from $S$.

		Next, we show that $A$ contributes at most $O(N+MP)$ to the total contention. Let $S$ be the set of processes that perform a single \acquire() operation and let $T$ be the set of processes that perform more than one \acquire() operation. If $p_k \in S$, then by Lemma \ref{lem:cas_ub}, there are at most $8$ CAS instructions on $A[k]$, so the total amount of contention on $A[k]$ is at most $8^2$. Let $a_k$ be the number of \acquire() operations performed by process $p_k$. If $p_k \in T$, then by Lemma \ref{lem:cas_ub}, there are at most $8a_k$ CAS instructions on $A[k]$. Let $b_k$ be the number of \release() operations performed by process $p_k$. Since each \acquire($k$) operation is always followed by a \release($k$) and since operation $a_k > 1$, we know that $b_k \geq a_k/2$. Therefore there are at most $16b_k$ CAS instructions on $A[k]$. Each CAS instruction causes at most $P$ contention, so the total contention on $A[k]$ is $16b_kP$. Therefore the total contention over all of $A$ is at most:

		\begin{align*}
			\sum_{p_k \in S} 8^2 + \sum_{p_k \in T} 16b_kP &= 8^2|S| + 16(\sum_{p_k \in T}b_k)P \\
&\leq 8^2N + 16MP \in O(N+MP)
		\end{align*}

		Finally, all we need to show is that \acquire() operations experience a total of $O(N)$ contention from accessing $A$. Again let $a_k$ be the number of \acquire() operations performed by process $p_k$. By Lemma \ref{lem:cas_ub}, there are at most $8a_k$ CAS instructions on $A[k]$. There can only be a single \acquire($k$) operation at a time, so each CAS on $A[k]$ contributes at most one unit of contention to at most one \acquire($k$) operation. Therefore \acquire($k$) operations experience at most $8a_k$ contention from $A[k]$. Summing over all $k$, we see that \acquire() operations experience a total of $O(N)$ contention from accessing $A$.
%
%
	\end{proof}

\end{maybeappendix}

    \section{Garbage Collection}\label{sec:gc}



In this section, we show how to efficiently collect out-of-date tuples on functional data structures
in the context of transactions and the VM problem.  We first define the desired properties of GC on functional data structures. We then present the \collect{} algorithm for our transactions (Figure \ref{alg:framework}) and show that it is fast and correct.

%
%
%

Intuitively, a linearizable precise VM solution provides an interface for safe and precise garbage collection over \emph{versions}, since \release[v] returns true if and only if it is the last usage of $v$. However, the precision and safety on the granularity of \emph{tuples} relies on a ``correct'' \collect{} \op{}, which, intuitively, should free all \tuple{s} that are no longer reachable as soon as possible. We formally define the desired property of a \emph{correct} \collect{} \op{}.

\begin{definition}\label{def:collect}
	Let $u$ be a \tuple, and $t$ be any time during an execution.
	A \collect{} is \emph{correct} if the following conditions hold.
	\begin{itemize}[noitemsep,topsep=0pt]
		\item If for each version $v$ that $u$ belongs to, \collect(v) has terminated by time $t$, then $u$ has been freed by $t$.
		\item If there exists a version $v$ that $u$ belongs to for which \collect(v) has not been called by time $t$, then $u$ has not been freed by $t$.
	\end{itemize}
\end{definition}



\myparagraph{The \collect{} Algorithm.} We now present a \collect{} algorithm and show its correctness and efficiency.
Path-copying causes subsets of the \tuple{s} to be shared among versions.
To collect the correct tuples, we use \emph{reference counting} (RC)~\cite{Collins60,Jones11} for enabling safe garbage collection. Each object maintains a count of references to it, and when it reaches $0$, it is safe to collect.
Since we use a PLM, the memory graph is acyclic. This means that RC allows collecting everything~\cite{Jones11}.
In our model, we maintain reference counts for each tuple $x$, \texttt{x.ref}, which records the number of ``parents'' of a node $x$ in the memory graph. Accordingly, a \texttt{tuple}$()$ operation creating a tuple $x$ increments the reference counters of all children of $x$.
We note that \texttt{tuple} can be called only by the writers' user code when it copies a path.
The counts are incremented only by the writers, but can be decreased by any \release{} \op{}.
A newly-created tuple $u$ has counter $0$. 
Later, when a transaction (reader or writer) executes a \collect{} of a version starting from \tuple$(x)$, it first decrements the count of $x$. Only if the count of $x$ has reached zero, $x$ gets freed, and all children of $x$ are collected recursively.
If $x$'s counter is more than one, the \collect{} \op{} terminates since the counts of its descendants will not be decreased then.

Pseudocode for \texttt{nth}$()$, \texttt{tuple}$()$ for a PLM, and the \collect$()$ \op{} is given in Algorithm \ref{alg:collect}. We use an array of length $l$ in each tuple $x$ to store the $l$ elements in this tuple (\texttt{x.ch[]}).
\texttt{inc} and \texttt{dec} denote atomic increment and decrement operations. We leave this general on purpose. The simplest way of implementing the counters is via a fetch-and-add object.
However, we note that this could introduce unnecessary contention. To mitigate that effect, other options, like dynamic non-zero indicators \cite{acar2017contention}, can be used.

 The result of this section is summarized in Theorem~\ref{thm:collect-correct}.

\begin{theorem}
	\label{thm:collect-correct}
	Our \collect{} algorithm (Algorithm \ref{alg:collect}) is correct and takes $O(S+1)$ time where $S$ is the number of \tuple{}s that were freed.
\end{theorem}

We show the proof of Theorem \ref{thm:collect-correct} in Appendix \ref{sec:waitfree-correctness}.
Intuitively, this is because tuples have a constant number of pointers and we only
recursively collect
any of those pointers if we free the tuple (the count has gone to zero).   We can therefore charge
the cost of visiting the child against the freed parent.

\begin{maybeappendix}{collect-correctness}
We now show that the \collect{} algorithm is correct. First we prove that it satisfies the first part of Definition \ref{def:collect}.

\begin{lemma} \label{lem:collect-1}
Let $u$ be a shared \tuple{}. For any shared \tuple{} $w$, let $V_w$ be the set of versions that $w$ belongs to. If a \collect{} operation has terminated for each version in $V_u$, then $u$ has been freed.
\end{lemma}

\begin{proof}
Fix an execution history and a configuration $C$. Consider the set $G$ of all shared \tuple{}s $w$ such that for each version $v \in V_w$, a \collect(v){} operation has terminated. 
It suffices to show that for each \tuple{} in $G$, there is a \collect{} operation that frees the tuple and terminates before $C$.

First, we show that no local \tuple{s} can affect the \tuple{s} of $G$. To see this, fix a \tuple{} $u \in G$. We want to show that there cannot be any pointers to $u$ from local \tuple{s}, and thus that its reference count cannot be affected by local \tuple{}s. Assume by contradiction that there is a local \tuple{} $\ell$ that is pointing to $u$ in configuration $C$. Note that only write transactions ever create \tuple{s}, and that the writer cleans up local \tuple{s} in its \texttt{output} \op{}, and therefore never leaves any local \tuple{s} or effect on the reference counts of shared \tuple{s} after returning.
Therefore, $\ell$ must have been created by a write transaction $t$ that is currently running user code. For $t$ to be able to create a \tuple{} that points to $u$, there are two cases: (1) $u$ must be a part of the version that $t$ commits, or (2) $u$ must be reachable from the version that $t$ acquired.
Note that in the first case, $u$ is not a shared \tuple{} itself, since it has been created by a transaction that has not yet finished its user code.
For the second case, recall that for $u$ to be in $G$, all versions that $u$ belongs to must have been collected. However, $u$ belongs to $V(t)$, and since $t$ is running user code, $V(t)$ is live at $C$, and therefore cannot have been collected yet. This contradicts the definition of $G$.
Therefore, $\ell$ cannot exist.

Notice that $G$ forms a DAG. Furthermore, for each \tuple{} $w \in G$, $G$ contains every shared \tuple{} that points to $w$. This is because a \tuple{} belongs to all of the versions that its parent belongs to. Therefore we can proceed by structural induction on $G$.

For the base of the induction, we prove that each of the roots in $G$ has been freed by a completed \collect{} operation. Let $u$ be some root in $G$.
We just need to show that each increment of $u$'s reference count has a completed \collect($u$) operation corresponding to it. We've already shown that there are no outstanding increments from local \tuple{s} affecting $u$.
This also holds for increments by \outputt($u$) operations because all of the versions that $u$ belongs to have already been collected. Since $u$ is a root, its reference count is not incremented anywhere else, so one of the completed \collect($u$) operation sets the reference count of $u$ to $0$ and frees $u$.

Now we prove the inductive step by fixing some \tuple{} $u$ in $G$ and assuming that all of its parents have been freed by some completed \collect{} operation. Similar to the base case, we show that each increment of $u$'s reference count has a completed \collect($u$) operation corresponding to it.
All arguments from the base case hold here, and therefore we do not need to worry about increments from local \tuple{s} or \outputt{} \op{}s.
So we just need to show that for each shared \tuple{} $w$ that point to $u$, there is also a completed \collect($u$) \op{}. By the inductive hypothesis, there is a completed \collect{} operation that frees $w$, and we can see from the code that this operation executes a \collect{} on $u$. Therefore one of the completed \collect($u$) operation sets the reference count of $u$ to $0$ and frees $u$.
By structural induction, each tuple in $G$ has been freed and this completes the proof.
\end{proof}

Next we prove that our collect algorithm satisfies the second part of Definition \ref{def:collect}.

\begin{lemma} \label{lem:collect-2}
Let $u$ be a shared \tuple{} and let $V_u$ be the set of versions that it belongs to. If a \collect{} operation has not started for some version $v \in V_u$, then $u$ has not been freed.
\end{lemma}

\begin{proof}

	Next we claim that each \collect($u$) operation corresponds to an unique increment of $u$'s reference counter.
	This can be seen by a close inspection of the code; let $c$ be a \collect($u$) call and consider two cases. Case (1): $c$ is not called from inside another \collect. That is, $u$ is the root of a version that is being collected. In that case, $c$ corresponds to the increment of $u.\mbox{\it ref}$ in the \outputt{} \op{} of the write that committed this version. Case (2): $c$ is called recursively from a \collect($u'$) \op{}. In this case, the $c$ corresponds to the increment of $u.\mbox{\it ref}$ during the creation of $u'$.

	Let $v \in V_u$ be the version for which no \collect($v$) call has been invoked.
	Since $u$ belongs to $v$, there must be a path from $v$'s version root $r$ to $u$ in the memory graph. We show by induction that no \tuple{} along that graph has been freed, thus implying that $u$ has not been freed.
	
	\textsc{Base:} Consider $v$'s root, $r$. $r.ref$ has been incremented by the \texttt{output} call of the writer that created the version $v$ and the \collect($v$) operation corresponding to this increment has not been invoked yet. Therefore the reference count of $r$ is non-zero, so it has not been freed.
	
	\textsc{Step:} Assume that the $i$th \tuple{}, $u_i$ in the path from $r$ to $u$ is not freed. We want to show that the $i+1$th \tuple{} on this path, $u_{i+1}$ has not been freed either.
	Consider the \texttt{\tuple} \op{} that made $u_i$ the parent of $u_{i+1}$ in the memory graph. That \op{} incremented $u_{i+1}$'s reference count by $1$ and the \collect($v$) operation corresponding to this increment has not been invoked yet because $u_i$ has not been freed.
	Thus, $u_{i+1}$'s reference count is greater than $0$, and therefore it cannot have been freed.
\end{proof}

Finally, we prove that our \collect{} algorithm is efficient.

\begin{lemma} \label{lem:collect-3}
A \collect{} operation takes $O(S+1)$ time where $S$ is the number of \tuple{}s that were freed by the operation.
\end{lemma}

\begin{proof}
Not counting the recursive calls, each \collect{} operation needs a constant time. Each time a \tuple{} is freed, a \collect{} operation is called on each of its $l$ children. Therefore, the total number of \collect{} operations spawned by a \collect{} operation $C$ is $l\times S$, where $S$ is the number of \tuple{}s that were freed by $C$. Since $l$ is constant, $C$ has $O(S+1)$ time complexity in total.
\end{proof}

Together, Lemmas \ref{lem:collect-1}, \ref{lem:collect-2} and \ref{lem:collect-3} imply Theorem \ref{thm:collect-correct}.
\end{maybeappendix}

\hide{
We note that the reference count of a tuple is never updated when the user code reads it, in this case, no contention is experienced, i.e.,

\begin{observation} \label{no-contention}
The \phasetwo{} of a read transaction experiences no contention.
\end{observation}
}


\renewcommand{\figurename}{Algorithm}
\lstset{basicstyle=\footnotesize\ttfamily, tabsize=2, escapeinside={@}{@}}
\begin{figure}
\hspace{-.27in}
\begin{minipage}{.57\columnwidth}
\begin{lstlisting}[numbers=none,basicstyle=\footnotesize\ttfamily]
var = int or Tuple
struct Tuple {
	var* ch[l]; int ref;}

Tuple tuple(var* x) {
	Tuple y=alloc(Tuple);
	y.ref=0;
	for (int i=0;i@<@l;i++){
		y.ch[i]=x[i];
		if (x[i] is Tuple)
			inc(x[i].ref);}}
void nth(Tuple x, int i){
    return x.ch[i];}
\end{lstlisting}
\end{minipage}
\hspace{-.15in}
\begin{minipage}{.52\columnwidth}
\begin{lstlisting}[numbers=none,basicstyle=\footnotesize\ttfamily]
void collect(var x) {
	if (x is int)
		return;
	int c=dec(x.ref);
	int i;
	if (c@$\le $@1) {
		var* tmp[l];
		for (i=0;i@<@l;i++)
			tmp[i]=nth(x, i);
		free(x);
		for (i=0;i@<@l;i++)
			collect(tmp[i]);
	}}	
\end{lstlisting}
\end{minipage}
%
%
%
\vspace{-.2in}
	\caption{\vspace{-.2in}\texttt{\tuple{}} and \collect{} algorithms}
\label{alg:collect}
\end{figure}

	\section{Implementing Transactions}
\label{sec:setup}

We now present our transaction system, and show that by plugging in
our \VerMain{} algorithm and underlying functional data structures
with correct GC, we can get an effective and efficient solution.  Read
and write transactions are implemented as shown in Figure
\ref{alg:framework}.  We assume all user code works in the functional
setting as described in Section~\ref{sec:prelim}.  The user code takes
in a pointer to a version root $v$, and may access (but not mutate)
any memory that is reachable from $v$.  The writer uses
path-copying, as standard in functional data structures, to construct
a new version.  It then can commit the version with the \set{}
operation.  Here we assume that the write transaction retries if the
\set{} fails (i.e., another concurrent write transaction has
succeeded).  Importantly the user code is unchanged from the
(functional) sequential code.  A read transaction is active until the last
instruction of its user code, and a write transaction is active until
the linearization point of its successful \set{} operation.
Transactions are live until the last instruction (after the \release{}
and GC).

\subsection{Correctness and Preciseness}

An instantiation of this framework consists of two important parts:
(1) a linearizable solution, $\vm$, to the version maintenance problem
defined in Section \ref{sec:vermain}, and (2) a correct \collect{}
function.  We show that combining them together yields strict
serializability, and safe and precise GC.

\begin{theorem}[Strictly Serializable] \label{thm:serial}
  Given a linearizable solution to the version maintenance problem,
  our transactional framework is strictly serializable.
\end{theorem}

For proving Theorem \ref{thm:serial}, we define a \emph{serialization point} for each transaction that is within its execution interval.

\begin{definition}\label{def:serpts}
	The serialization point, $s$, of a transaction $t$ is:
	\begin{itemize}
		\item If $t$ is a read transaction, then $s$ is at the linearization point of $t$'s call to $\vm.$\acquire[]().
		\item If $t$ is a write transaction, then $s$ is at the linearization point of $t$'s call to its successful $\vm.$\set[]().
	\end{itemize}
\end{definition}

A proof is given in \cite{singlewriterarxiv}.
Intuitively,
we show that if we sequentialize any given history according to these
serialization points, it is equivalent to some sequential
transactional history.

\begin{theorem}[Safe and Precice GC] \label{thm:safe-precise}
  Given a linearizable solution to the version maintenance problem and
  a correct \collect{} function, our garbage collection is safe and
  precise.
\end{theorem}

A full proof is given in Appendix~\ref{app:theorem2}.
Intuitively,
the garbage collection is safe because \collect($v$) is called only
when a \release[] returns $v$, meaning that $v$ is no longer
live.  It is precise since if the \release{} is the last one on the
transaction's version, the precise \VerMain{} solution will return that version,
and any tuples in the version that are not shared with other versions
will be collected while the transaction is still live.  Therefore no
version that is no longer live will survive past the lifetime of the
last transaction that releases it.

\lstset{basicstyle=\footnotesize\ttfamily, tabsize=2, escapeinside={@}{@}}
\lstset{literate={<<}{{$\langle$}}1  {>>}{{$\rangle$}}1}
\lstset{language=C, morekeywords={CAS,commit,empty,local,job,taken,entry,GOTO}}
\lstset{xleftmargin=5.0ex, numbers=left, numberblanklines=false, frame=single}
\makeatletter
\lst@Key{countblanklines}{true}[t]%
{\lstKV@SetIf{#1}\lst@ifcountblanklines}

\lst@AddToHook{OnEmptyLine}{%
	\lst@ifnumberblanklines\else%
	\lst@ifcountblanklines\else%
	\advance\c@lstnumber-\@ne\relax%
	\fi%
	\fi}
\makeatother

\hide{
\begin{figure}
	\caption{Transaction Phases. $k$ is the processor id.\vspace{-.2in}}
	\begin{minipage}{.51\textwidth}
\begin{lstlisting}[caption={Read Transaction}]
v = M.acquire(k);	  //@\phaseone@
user_code(v);		  	//@\phasetwo@
//Respond to user here
res = M.release(k);	//@\phasethree@
if (res) {			   	//@\phasefour@
	collect(v);	}		
\end{lstlisting}
	\end{minipage}\hfill
	\begin{minipage}{0.51\textwidth}
\begin{lstlisting}[caption={Write Transaction}]
v = M.acquire(k);		//@\phaseone@
new = user_code(v);	//@\phasetwo@
M.set(new);				  //@\phasethree@
res = M.release(k);	
if (res) {					//@\phasefour@
	collect(v);	}
\end{lstlisting}
	\end{minipage}
	\label{alg:framework}
\end{figure}}


\begin{maybeappendix}{serial-proof}
\begin{lemma}
	Let $H$ be a transactional history, and $w$ a write transaction that commits version $v_w$ in $H$.
	Let $S$ be a serialization of $H$ according to the serialization points outlined in Definition~\ref{def:serpts}.
	A read transaction $r$ uses $v_w$ as its version if and only if $w$ is the last write transaction before $r$ in $S$.	
\end{lemma}

\begin{proof}
	In our framework, a transaction always uses the version returned by its call to $\vm.$\acquire[]().
	By the definition of the \VerMain{} problem, the $\vm.$\acquire[]() \op{} returns the current version of $\vm$ at the time that the \acquire[] is linearized. Thus, $r$ uses the version that is current in $\vm$ at the time that it serializes (since its serialization point is the same as the linearization point of its call to \acquire[]).
	Recall that the current version of a \VerMain{} instance is by definition the version that was set by the most recent \set[] \op{}.
	Note that in the transactional history, the only calls to $\vm.$\set[] are from write transactions, and each write transaction serializes at the linearization point of its only \set[] \op{}.
	Thus, if the read transaction, $r$, uses version $v_w$, $v_w$ must have been the current version at $r$'s serialization point. Since only write transactions call $\vm.$\set[], and they serialize at the linearization point of this \set[] \op{}, by definition of the current version, $w$ must have been the last write transaction serialized before $r$.
\end{proof}

To complete the proof of serializability, we also need to show that the write transactions are atomic, i.e., that the current version never changes between when the write transaction acquires a version and when it commits a new version. However, this trivially holds, since we do not allow concurrent write transactions.
Thus, we conclude Theorem \ref{thm:serial}.
\end{maybeappendix}


\hide{
\paragraph{Extension to Multi-writer scenario.} We discuss methods to extend our approach to multiple writers in the supplement material. In our experiments, we get around this by batching concurrent updates \cite{hendler2010flat,batcher}.
}


\begin{maybeappendix}{gc-correctness}
We now turn to showing that garbage collection can be safe and precise under this framework.

\begin{theorem}
	If we have a correct \collect{} \op{}, then garbage collection is precise.
\end{theorem}
\begin{proof}
	Recall that garbage collection is said to be \emph{precise} if the allocated space at every point in the execution is a subset of the reachable space at that point. 
	Consider some \tuple{} $u$ that becomes unreachable at time $t$. We must show that $u$ is freed by $t$.
	
	By definition, a version $v$ is live at a configuration $C$ if and only if either (1) $v$ is the current version at $C$, or (2) $v$ has been acquired but not yet released. Note that if the second property holds, then there is a transaction $t$ for which $V(t) = v$. Thus, if $v$ is live, then its root is also reachable (and so are all \tuple{}s that belong to $v$).
	
	Let $V_u$ be the set of versions that $u$ belongs to.
	From the above, we know that for all versions $v \in V_u$, $v$ must not be live by time $t$, since otherwise $u$ would still be reachable, contradicting the  definition of $u$.
	This means that for all $v \in V_u$, some \release[v]() \op{} has returned \true{} (by the \VerMain{} problem definition), and thus the transaction that executed this \release[v]() called \collect(v).
	
	Assume by contradiction that by time $t$, there is some $v\in V_u$ for which the call to \collect(v) has not yet terminated. Let $\tau$ be the transaction whose \release[v]() \op{} returned \true{}. Since $\tau$'s \collect(v) hasn't yet terminated, $\tau$ is still \activet{}. But then $u$ is reachable from an \activet{} transaction, $\tau$, contradicting the definition of $u$.
	Therefore, for all $v\in V_u$, \collect(v) has terminated by time $t$.
	By the correctness of the \collect{} function, $u$ must have been freed by $t$.
	%
\end{proof}

\begin{theorem}
	If we have a correct \collect{} \op{}, then garbage collection is safe.
\end{theorem}

\newcommand{\cal}[1]{{\mathcal{V}(#1)}}
\begin{proof}
	Recall that garbage collection is said to be \emph{safe} if at any point in an execution, the allocated space is a superset of $R(T_{early})$, where $T_{early}$ is the set of transactions that have not finished executing user code yet, and $R(T)$ is the reachable space from a set of transactions $T$. 
	
	Let ${\cal{V}}(T_{early})$ be the set of versions of transactions in $T_{early}$. By the sequential specification of the \VerMain, all versions in ${\cal{V}}(T_{early})$ are still live (since they've been acquired but not released).
	This means that no \release[v]() \op{} has returned \true{} for any version $v \in {\cal{V}}(T_{early})$. Thus, for every such version $v$, \collect(v) has not been called so far in the execution.
	
	Thus, from the correctness of the \collect() function, no \tuple{} that is reachable from any such version $v$ has been freed.
	%
\end{proof}

\end{maybeappendix}


\subsection{Delay and Contention}
Here we prove bounds on delay and contention experienced by
transactions assuming we use the wait-free algorithm for the version
maintenance problem (Section \ref{sec:waitfree}), and our \collect{}
function (Section \ref{sec:gc}).  A summary of the results is shown in
Table \ref{tab:allresults}.

\begin{theorem}[Step Complexity]
  \label{thm:transaction-delay}
  With our transactional system using the \PSWF{} algorithm for
  \VerMain{},
  \begin{itemize}
  \item all read transactions are delay-free,
  \item all single-writer transactions have $O(P)$ delay, and
  \item all write transactions are lock-free.
  \end{itemize}
  Furthermore, for single-writers, the time complexity of the garbage
  collection across a sequence of transactions is bounded by the
  number of unique tuples used across all versions.
\end{theorem}

\begin{proof}
  The proof follows almost directly from previous
  theorems~\ref{thm:waitfree-step} and \ref{thm:collect-correct}.  In
  particular, a read-transaction is active during the \acquire{} and
  the user code.  The \acquire{} takes $O(1)$ time by Theorem
  \ref{thm:waitfree-step}, and the user code requires no extra time since the
  code is not changed from the original sequential code.  The
  transaction is therefore delay-free.  A write transaction is active
  during the \acquire{}, user code and until the end of a successful
  \set{}.  The cost of \acquire{} is $O(1)$, the cost of
  \set{} is $O(P)$ and the user code takes no more
  time than it would sequentially.  If there is no concurrent writer
  it will succeed on the first try and hence have delay $O(P)$.  If
  concurrent with other writers it can only fail and restart if some other writing transaction succeeds.  Hence it is lock-free. 

  In the single-writer context, all values are successfully written and
  hence the number of tuples needed to collect is bounded by the tuples
  that appear across all versions.  By Theorem \ref{thm:collect-correct} each
  takes constant time to collect.
  \end{proof}

\begin{theorem}
  \label{thm:transContention}
  For the single-writer setting, all shared-memory operations except inside the garbage
  collector have $O(1)$ amortized contention.
\end{theorem}
\begin{proof}
  This follows the bounds on contention in
  Theorem~\ref{thm:contention} for \acquire{}, \set{}, and \release{}.
  Each has amortized contention proportional to its time complexity.   Furthermore in the single-writer context, only a single transaction is allocating and incrementing reference counts at any time.   However, in the garbage collection there can
  be contention when decrementing reference counts.
  \end{proof}


\subsection{Discussion about Functional Data Structures}

The important features of the functional code for our purposes is that it
is fully persistent and safe for concurrency, both by default.  As
previously mentioned, persistence can also be achieved by using
version lists on each
object~\cite{Reed78,BG83,papadimitriou1984concurrency,Kumar14,neumann2015fast}.
This requires modifying every read and write, and can asymptotically
increase the time complexity of user code.  There has been theoretical work on
efficiently supporting version-list based persistence based on node
splitting~\cite{driscoll1989making}.  This approach, however, has
several drawbacks in our context.  Firstly it requires at most a
constant number of pointers to all objects.  This would disallow, for
example, even having the processes point to a common object.
Secondly, it is not safe for concurrency.  Making it safe would be an
interesting and non-trivial research topic on its own.  Thirdly, the
approach does not address garbage collection---it assumes all versions
are maintained.  Again, adding garbage collection would be an interesting
research topic on its own.  Finally, constant time operations are only
supported for what is called partial persistence---i.e. a linear
history of changes.  Supporting lock-free writers seems to require
that multiple writers simultaneously update their versions, which
requires what is called full persistence, which allows for branching
of the history.

We note that a disadvantage of functional data structures as compared
to version lists is that they sequentialize write transactions even
when on different parts of a data structure.  With version lists, if
two transactions are race-free (the set of objects that one writes is
disjoint from the set that the other reads and writes), then they can
proceed in parallel and serialize in either order.  For this
reason, we believe our approach is best suited either in situations
when the transaction load is dominated by readers, or when the updates
can be batched, as described in our experiments.  As mentioned in the
introduction, due to dependences it is impossible to bound the delay
for writers independently of the other concurrent writers.  It might
be possible, however, to bound delays relative to inherent
dependences---i.e., the delay is no more than forced by a dependence.

\begin{maybeappendix}{discussion}
\section{Discussion}\label{sec:discussion}

Some of the techniques in our algorithm can also be found in wait-free universal
construction algorithms \cite{fatourou2011highly, herlihy1990methodology,herlihy1993methodology}). 
Most universal constructions tend to be impractical because they copy the state of
the data structure for each new operation.
Viewed from the universal construction perspective, we presented a single-writer
universal construction algorithm that (1) does not use large registers, (2) reduces
amount of variables copied by using functional data structures and path copying, (3) special cases read
operations so that they do not have to copy and (4) garbage collects old versions
in a precise manner.
The last point in particular is interesting because we have not seen any other universal
construction algorithms with precise garbage collection and this is the problem
that our Version Maintenance Problem is designed to address.
\end{maybeappendix}

	\section{Other VM Algorithms}
\label{sec:vm_algs}

In this section, we present three additional solutions to the Version Maintenance problem. One solution is based on Read-Copy-Update RCU~\cite{McKenney98} and the other two are based on widely used memory reclamation techniques: Hazard Pointers (HP) \cite{michael2004hazard} and Epoch Based Reclamation (EP) \cite{fraser2004practical}.
These solutions are simple to describe,
but have various drawbacks. The HP and EP based solutions are not precise. RCU leads to a precise solution, but writers block waiting for readers.
Researchers have proposed numerous extensions to the original HP and EP techniques~\cite{aghazadeh2014making, cohen2015efficient, wen2018interval, brown2015reclaiming}. Some of these directly translate to new ways of solving the VM problem. Our \PSWF{} algorithm can be understood as a wait-free and precise extension of the HP based algorithm.
We experimentally compare these version maintenance strategies in Section \ref{sec:exp:rangesum}.

\vspace{1mm}
\myparagraph{Read-Copy-Update (RCU).}
The basic RCU interface provides 3 methods: \rlock{}, \runlock{}, and \rsync{}. \\ \rlock{} and \runlock{} mark the beginning and end of read-side critical sections. \rsync{} blocks until all the currently active read-side critical sections have completed. Note that \rsync{} only needs to wait for the read-side critical sections that existed at the start of its execution.

The RCU-based \acquire{} method calls \rlock{} and then reads and returns the current version. The \set{} method updates the current version using a CAS (similar to the \PSWF{} algorithm). If the CAS succeeds, it remembers the old version. If \release{} does not follow a successful \set, it simply calls \runlock{} and returns the empty set. Otherwise, it also has to call \rsync{} and return the old version to be garbage collected. The downside of RCU is that write transactions have to wait for read transactions which led to slow write throughput in our experiments. We use the Citrus \cite{arbel2014concurrent} implementation of RCU for our experiments.

\vspace{1mm}
\myparagraph{Hazard Pointers (HP).}
To \acquire{} a version in the HP based algorithm, a process $p$ first reads the current version and announces it. This announcement tells other processes that the version is potentially being used. Then $p$ reads the current version again to check if it has changed. If not, then the announced version was still current at the time of the announcement and $p$ can safely return the version it announced. Otherwise, the \acquire{} has to restart. A \set{} operation simply updates the current version using a CAS, and if the CAS succeeds, it adds the old version to its retired list. A \release{} operation by $p$ first clears its announcement location and if its retired list reaches size $2P$, it scans the announcement array and it removes and returns all the versions in its retired list that were not announced.
Any version retired by $p$ that is was not announced is safe to collect because it cannot be returned by a future \acquire{} operation; it might be announced by a future \acquire{}, but that operation would detect that the current version has changed and restart.
If the retired list has size $2P$, then the \release{} operation returns at least $P$ versions and can be implemented using $O(P)$ time. Otherwise, the \release{} operation returns an empty list and takes $O(1)$ time. There are at least $P$ fast \release{} operations between each expensive one so its amortized time complexity is $O(1)$. Note that \release{} always returns an empty list for read-only processes.

\vspace{1mm}
\myparagraph{Epoch Based Reclamation (EP).}
In EP, the execution is divided into epochs and for each epoch, we maintain the set of versions that were retired during that epoch. An \acquire{} operation simply reads and announces the current epoch, and then reads and returns the current version. A \release{} operation reads the current epoch and scans the announcement array.
If everyone has announced this epoch, it tries to increment the current epoch with a CAS. If the CAS succeeds, it returns all the versions retired 2 epochs ago. Since everyone has announced the previous epoch, these versions cannot be accessed anymore. In all other cases, the \release{} operation returns an empty list. It is only necessary to maintain a set of retired versions for the last 3 epochs.

To reduce the number of times we scan the announcement array, we only do this for \release{} operations that follow a successful \set{} operation. All other \release{} operations are allowed to return right away. This optimization increases the number of uncollected versions by at most 1.

	\renewcommand{\figurename}{Figure}
\section{Experiments}
\label{sec:exp}
In this section, we study the performance of our approach using ordered
maps implemented with balanced binary trees.  For the ordered maps we
use the C++ PAM library~\cite{pam} since it already supports
functional tree structures, and has a reference counting collector.
For the experiments, we have implemented five
versions of the \VerMain{}: our \PSWF{} algorithm, our algorithm without helping, an
imprecise version based on epochs, an imprecise version based on
hazard pointers, and a blocking version based on RCU. We do not compare to general purpose software transactional memory systems
since previous results show they are not competitive to direct concurrent
implementations~\cite{Gramoli15}.

We run two types of experiments.  The first studies query and
update operations under a single-writer multi-reader concurrent
setting.  The experiments are designed to understand the overheads of
the different \VerMain{} algorithm and how much
garbage they leave behind.  The
second type measures the throughput of concurrent operations on
functional trees, comparing
to five existing trees (or skiplists).
It uses batching for our functional
tree structure.  The
goal is to understand the overhead of using functional trees.

\myparagraph{Setup.}  For all experiments, we use a 72-core Dell R930 with
4 x Intel(R) Xeon(R) E7-8867 v4 (18 cores, 2.4GHz and 45MB L3 cache),
and 1Tbyte memory.  Each core is 2-way hyperthreaded giving 144
hyperthreads.  Our code was compiled using g++ 5.4.1 with
the Cilk Plus extensions.  We compile with \texttt{-O3}. We use
\texttt{numactl -i all} in all experiments, evenly spreading the
memory pages across the processors in a round-robin fashion.
All the numbers are taken by averaging of 3 runs. In experiments, we use ``threads'' to refer
to ``processes'' as we use in our theoretical analysis.

\subsection{Evaluating the VM Algorithms and GC}
\label{sec:exp:rangesum}
In this section, we experiment with five different \VerMain{}
algorithms: our precise, safe and wait-free algorithm from Section~\ref{sec:vermain} (\emph{\PSWF}), our algorithm without helping (which only guarantees lock-freedom, referred to as \emph{PSLF}),
a hazard-pointer-based algorithm (\emph{HP}), an epoch-based algorithm (\emph{EP}), and an RCU-based algorithm (\emph{RCU}). The implementation of the latter three is discussed in Section \ref{sec:vm_algs}. We note that \PSWF{}, PSLF and RCU guarantee precise garbage collection, while EP and HP do not. RCU guarantees that at any point there are at most two live versions, but will block writers if there are readers working on the old version. HP, EP, and our \PSWF{} algorithm are non-blocking.

We use the functional augmented tree structure in PAM as the underlying data structure.
We use integer keys and values, and conduct parallel range-sum queries while
updating the tree with insertions. Each query asks for the sum of values
in a key range in time $O(\log n)$ with augmentation.
The initial tree size is $n=10^8$. We use $P=141$ threads to
invoke concurrent transactions, among which one thread continually
commits updates, each containing $n_u$ sequential insertions, and
$140$ threads conduct queries, each containing $n_q$ range-sum
queries.  We control the granularity of update and query transactions
by adjusting $n_u$ and $n_q$, respectively.  We set the total running time to be 15
seconds, and test different combinations of update and query granularity.
We keep track of the number of live versions before each update, and report the maximum number of versions. 
The results are shown in Table \ref{tab:vmexp} and Figure \ref{fig:max_live}.

\begin{fullversion}
\begin{table}[!h!t]
  \centering
    \begin{tabular}{|r|r|r|r|r|r|}
    \hline
    \multicolumn{1}{|c|}{{\textbf{Query}}} & \multicolumn{5}{c|}{\textbf{Update Granularity}} \\
\cline{2-6} \multicolumn{1}{|c|}{{\textbf{Granularity}}}          & \multicolumn{1}{c|}{\textbf{1}} & \multicolumn{1}{c|}{\textbf{10}} & \multicolumn{1}{c|}{\textbf{100}} & \multicolumn{1}{c|}{\textbf{1000}} & \multicolumn{1}{c|}{\textbf{10000}} \\
    \hline
    \textbf{1} & 8     & 4     & 3     & 2     & 1 \\
    \hline
    \textbf{10} & 9     & 4     & 3     & 2     & 1 \\
    \hline
    \textbf{100} & 42    & 7     & 3     & 2     & 1 \\
    \hline
    \textbf{1000} & 113   & 54    & 9     & 2     & 1 \\
    \hline
    \textbf{10000} & 141   & 130   & 55    & 8     & 2 \\
    \hline
    \end{tabular}%
      \caption{The maximum number of live versions.}
  \label{tab:maxversion}%
\end{table}%

\begin{table}[!h!t]
  \centering
    \begin{tabular}{|r|r|r|r|r|r|}
    \hline
    \multicolumn{1}{|c|}{\textbf{Query}} & \multicolumn{5}{c|}{\textbf{Update Granularity}} \\
\cline{2-6}    \multicolumn{1}{|c|}{\textbf{Granularity}} & \multicolumn{1}{c|}{\textbf{1}} & \multicolumn{1}{c|}{\textbf{10}} & \multicolumn{1}{c|}{\textbf{100}} & \multicolumn{1}{c|}{\textbf{1000}} & \multicolumn{1}{c|}{\textbf{10000}} \\
    \hline
    \textbf{1} & 1.92  & 1.29  & 1.16  & 1.00  & 1.00  \\
    \hline
    \textbf{10} & 3.77  & 1.21  & 1.21  & 1.00  & 1.00  \\
    \hline
    \textbf{100} & 25.47  & 4.34  & 1.22  & 1.00  & 1.00  \\
    \hline
    \textbf{1000} & 81.71  & 34.56  & 4.93  & 1.03  & 1.00  \\
    \hline
    \textbf{10000} & 133.55  & 109.30  & 40.76  & 5.51  & 1.02  \\
    \hline
    \end{tabular}%
      \caption{Average number of live versions. }
  \label{tab:aveversion}%
\end{table}%
\end{fullversion}

\paragraph{\textbf{The number of live versions.}}
The number of live versions for all five algorithms in different settings is shown in Table \ref{tab:vmexp}.
Figure \ref{fig:max_live} shows the maximum live versions of the five VM algorithms,
with different update granularity
when $n_q=10$.
The general trends for all five algorithms are similar.
When $n_u$ is large or $n_q$ is small, there are few versions live. This is
because when updates are less frequent or queries finish fast, most queries
will catch recent versions.
When $n_u$ is small or $n_q$ is large,
the number of live versions gets larger.  This is because when new versions
are generated frequently, or queries take a long time, it is more likely
for queries to be behind the current version, and keep more old versions live.

\begin{figure}[!t]
\centering
\begin{minipage}[t]{0.48\columnwidth}
\includegraphics[width=\columnwidth]{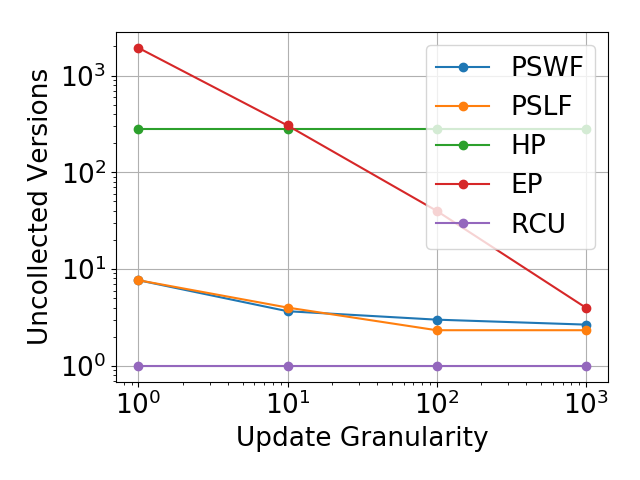}
\begin{minipage}[c]{0.9\columnwidth}
\vspace{-.1in}
\caption{\small Maximum number of uncollected versions for different VM algorithms. $n_q$ is 10, 140 query threads.\label{fig:max_live}}
\end{minipage}
\end{minipage}
\begin{minipage}[t]{0.48\columnwidth}
\includegraphics[width=\columnwidth]{figures/ycsb}
\begin{minipage}[c]{0.9\columnwidth}
\vspace{-.1in}
\caption{\small Throughput of six data structures on YCSB workloads A (read/update, 50/50), B (read/update, 95/5) and C (all reads). \label{fig:augsumexp}}
\end{minipage}
\end{minipage}
\vspace{-.15in}
\end{figure}


We now compare the five VM algorithms.
The maximum number of live versions for HP is always $2P=282$.
For EP, when $n_u$ is large, the number of live versions
is reasonable and mostly below 100. However, for frequent updates,
the number of versions can reach up to 1000 (see Figure \ref{fig:max_live}), because queries cannot
catch up with the latest version. Many recent (but not current)
versions cannot be collected, even if no queries are working on them.
Theoretically the epoch-based algorithm can leave an unbounded number of versions behind.
RCU keeps only 1 version before \set{}
since the writer will wait to collect the old version before generating a new version.
Although the amount of garbage is small, the writer is blocked and update granularity
is low as we will show later in this section.
For our \PSWF{} algorithm, the number of total versions is at most $141$ for small $n_u$ and large $n_q$.
This case is possible but rare to occur. In the settings we shown in this paper, the maximum number of versions is
within 100. In most of the cases, the maximum of living versions is around $10$, which is $1/14$ of the total query
threads.  Because our GC is precise, all out-of-date versions are
collected immediately. The helping scheme is our \PSWF{} does not affect much of the number of maximum versions.
For all tested setting, the number of versions
kept by our \PSWF{} algorithm is only 1.5-83$\times$ less than EP, and
about 7-120$\times$ less than HP.



\newcommand{\cmark}{\ding{51}}%
\newcommand{\xmark}{\ding{55}}%

\paragraph{\textbf{The throughput of queries and updates.}}
We report the query and update throughput (millions of queries/updates per second) for different settings in
Table \ref{tab:vmexp}. We compare the throughput
numbers for base cases when no VM (and thus no GC) algorithms are adopted, noted as ``Base''
in the Tables. 

\begin{table}[!t]
  \centering
  \small
    \begin{tabular}{rr|rrrrrr}
      $\boldsymbol{n_q}$ & \multicolumn{1}{c}{$\boldsymbol{n_u}$} & \multicolumn{1}{c}{\textbf{Base}} & \multicolumn{1}{c}{\underline{\textbf{\PSWF}}} & \multicolumn{1}{c}{\underline{\textbf{PSLF}}} & \multicolumn{1}{c}{\textbf{HP}} & \multicolumn{1}{c}{\textbf{EP}} & \multicolumn{1}{c}{\textbf{RCU}} \\
    \hline
     & & \multicolumn{6}{c}{\textbf{Query Throughput (Mop/s)}}  \\

        \textbf{10} & \textbf{10} & 44.40& 39.79& 39.51& 39.46& 39.07& 39.20 \\
        \textbf{10} & \textbf{1000} & 44.63& 39.40& 39.51& 42.31& 39.74& 39.55 \\
        \textbf{1000} & \textbf{10} & 46.24& 40.54& 40.53& 41.16& 40.29& 47.74 \\
        \textbf{1000} & \textbf{1000} & 46.22& 41.10& 40.56& 43.76& 40.94& 41.45 \\

    \hline
%
    & & \multicolumn{6}{c}{\textbf{Update Throughput (Mop/s)}}  \\
        \textbf{10} & \textbf{10} & 0.133& 0.101& 0.104& 0.053& 0.064& 0.056 \\
        \textbf{10} & \textbf{1000} & 0.158& 0.133& 0.134& 0.074& 0.071& 0.073 \\
        \textbf{1000} & \textbf{10} & 0.130& 0.105& 0.107& 0.056& 0.063& 0.003 \\
        \textbf{1000} & \textbf{1000} & 0.154& 0.133& 0.134& 0.077& 0.074& 0.060 \\
    \hline
%
& & \multicolumn{6}{c}{\textbf{Max \# Versions}}  \\
        \textbf{10} & \textbf{10} & ---& 3.67& 4.00& 282.00& 304.67& 1.00 \\
        \textbf{10} & \textbf{1000} & ---& 2.67& 2.33& 282.00& 4.00& 1.00 \\
        \textbf{1000} & \textbf{10} & ---& 36.33& 36.33& 282.00& 324.00& 1.00 \\
        \textbf{1000} & \textbf{1000} & ---& 2.33& 2.00& 282.00& 3.33& 1.00 \\
    \hline
    \end{tabular}%
  \caption{\small The query throughput, update throughput, and the number of live versions in each VM algorithm under various settings. Throughput numbers are reported as millions of operations per second (Mop/s).}\label{tab:vmexp}%
\end{table}%

Generally, from Table \ref{tab:vmexp} we can see that introducing a VM algorithm
always lowers the throughput of queries and updates. This is not only
because of the overhead in maintaining versions, but also from
the possible GC cost. For both updates and queries, we do not see a significant difference between
our \PSWF{} algorithm and PSLF algorithm. Generally this means that in practice, it is very rare that the writer
needs to help the readers a lot. We do see a more notable difference in extreme cases (e.g., $n_u=1$) \cite{singlewriterarxiv}.

\emph{Queries.} For all the five algorithms and all the four settings,
the overhead of introducing GC and VM algorithms is around 10\% for queries.
The five VM algorithms have comparable performance. RCU usually has much better query performance,
this is possibly because all the queries of RCU are working on the same version, and thus leading to
better locality.

\emph{Updates.} Generally, larger $n_u$ results in better update throughput. There are mainly two reasons. Firstly, batching more updates in one transaction reduces the overhead in calling \acquire{}, \set{} and \release{} for version maintenance. Secondly, larger update transactions allow more query threads to catch more recent versions, and thus a larger fraction of the current version will appear in cache, making updates faster. The overhead of introducing GC and VM algorithms is within 20\% for our \PSWF{} algorithm, but can be more for the other algorithms.
Our algorithms are always the best among all the algorithms in terms of update throughput.
It is likely because for HP, EP and RCU, the writer is responsible to do all GC work, while in \PSWF{},
queries and updates share the responsibility of GC. Note that although RCU has the best performance in queries,
it has much lower update performance than the others, because the writer can be blocked by unfinished queries.

\emph{Overall.} Generally, our \PSWF{} algorithm is comparable to
the EP and HP, and slightly slower than RCU in queries, but is always much faster in updates than all the other implementations. As mentioned, this is mostly due to the difference in
GC responsibility. 
Therefore, our algorithms have the best overall
performance.

\subsection{Functional Concurrent Operations}
In this section test the throughput of concurrent operations on
the functional tree in PAM.

\myparagraph{Concurrent Operations with Batching.}
We compare the functional tree to several state-of-the-art concurrent data structures:
skiplists~\cite{skiplist}, OpenBW trees~\cite{openbwtree}, Masstree
\cite{masstree}, B+trees \cite{openbwtree} and concurrent Chromatic
trees~\cite{Brown14,chromatic} (all in C++).
For all
structures we turn GC off since we are interested in the performance
of the trees and not the GC.
We use the Yahoo! Cloud Serving Benchmark
(YCSB) microbenchmarks, which have skewed access patterns (Zipfian
distributions) to mimic real-world access patterns.  We test YCSB
workloads A (read/update, 50/50), B (read/update, 95/5) and C (all
read).  The original dataset (before updates) has $5\times 10^7$
elements, and each workload contains $10^7$ transactions.  We use
64-bit integers.

For PAM we use batching to collect concurrent updates so they
can be updated in parallel using single-writer.  The batching works by accumulating update
requests in a buffer and when there are a sufficiently many, applying
them using PAM's multi-insert function, which is a parallel divide-and-conquer algorithm~\cite{blelloch2016just}.
The batch size is controlled so the latency for an update is no more than 50ms.  More
details on batching are given in Appendix \ref{sec:batching}.  The
reads (finds in the tree) do not need to be batched since any number
of readers can run concurrently.

The results on operation throughput are presented in Figure
\ref{fig:augsumexp}.  In all the three workloads, our
implementation outperforms the best of the others by 20\%-300\%.
There are a few factors
contribute to the good performance of our implementation.  Firstly,
the code for a query is just a standard tree search with no additional
cost for synchronization.  Secondly, since the code for the batched
updates uses a parallel divide-and-conquer algorithm for each batch,
it generates no contention between writes.

We note that the comparison is not apples-to-apples.  Due to batching,
our updates have higher
latency than the others.  This will
not be appropriate in some applications.  On the other hand, our
approach allows multiple operations to be applied atomically,
while the others only support
atomicity at the granularity of individual operations.


\hide{In summary, in situations in which updates can be batched, our
approach seems to outperform the throughput of other concurrent trees,
and in situations where atomicity of multiple-operation is required,
the other concurrent trees are not even directly applicable.}

\begin{fullversion}
\begin{table}[htbp]
  \centering
    \begin{tabular}{@{}r@{}l@{}||rrrrrr}
    \hline
          &       &       & \multicolumn{1}{c}{\textbf{skip-}} & \multicolumn{1}{c}{\textbf{Open-}} & \multicolumn{1}{c}{\textbf{mass-}} & \multicolumn{1}{c}{\textbf{B+}} & \multicolumn{1}{c}{\textbf{Chro-}} \\
          & \multicolumn{1}{c||}{\textbf{R/U}} & \multicolumn{1}{c}{\textbf{PAM}} & \multicolumn{1}{c}{\textbf{list}} & \multicolumn{1}{c}{\textbf{BW}} & \multicolumn{1}{c}{\textbf{tree}} & \multicolumn{1}{c}{\textbf{tree}} & \multicolumn{1}{c}{\textbf{matic}} \\
    \hline\hline
    \textbf{A} & \textbf{ (50/50)} & \textbf{27.62} & 0.92  & 8.77  & 14.57 & 16.39 & 19.19 \\
    \textbf{B} & \textbf{ (95/5)} & \textcolor[rgb]{ .133,  .133,  .133}{\textbf{102.21}} & 1.32  & 30.78 & 84.96 & 78.50 & 41.46 \\
    \textbf{C} & \textbf{ (100/0)} & \textbf{139.11} & 13.24 & 52.12 & 98.16 & 91.02 & 43.84 \\
    \hline
    \end{tabular}%
    \caption{Throughput of six concurrent data structures on YCSB workloads A (read/update, 50/50), B (read/update, 95/5)and C (all reads).}
  \label{tab:addlabel}%
\end{table}%
\end{fullversion}

\hide{
\begin{figure}[!h!t]
  \centering
  \includegraphics[width=1\columnwidth]{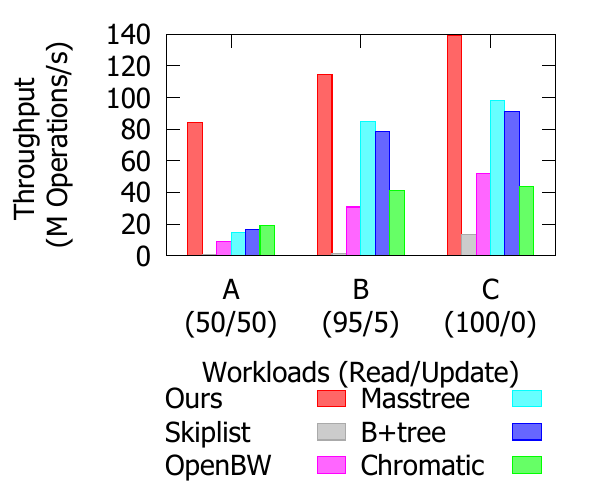}
  \caption{Throughput of six concurrent data structures on YCSB workloads A (read/update, 50/50), B (read/update, 95/5)and C (all reads).}\label{fig:ycsb}
\end{figure}
}
\myparagraph{Inverted Index Searching.}
We test the functional tree on searching an inverted index \cite{rajaraman2011,ZM06}
to show the overhead of read/write transactions on functional data structures. 

We also test our algorithm on searching a weighted inverted index \cite{ZM06,rajaraman2011} (also called an inverted file or posted file). Given a set of documents, which each consists of a set of terms (words), an inverted index build a mapping from each term to a list of documents (called the \emph{posting list} of the term) it appears in, each document assigned a weight $w$ corresponding to the term-document pair. Usually the weight reflects how the term is related to the document, and how important is the document itself. 

We implement the mapping using a tree $T$, where the value (the \emph{posting list}) of each term $t$ is a inner map structure, noted as $pl(t)$, mapping each document $d$ to a weight $w_{t,d}$. We augment the inner tree with the maximum weight in its subtree. Both the inner and the outer trees are functional using path-copying. The static setting is basically the same as in \cite{pam}. In this paper, we evaluate the throughput of the tree in the dynamic setting, i.e., when updates and queries are done concurrently.

\hide{
We use the augmented map framework to represent $M_I$ and $M_O$ as follows:
{\small
\begin{tabular}{l@{}@{ }l@{ }@{}l@{}@{ }l@{ }@{}l@{ }@{}l@{ }@{}l@{ }@{}l@{ }@{}l@{ }@{}l@{ }@{}l@{ }@{}l@{ }}
$M_I$ &$=$& $\mm$&$($&$D$,&$W$,&$<_D$,&$W$,&$(k,v) \mapsto v$, &${\max}_W$,&$0$&)\\
$M_O$ &$=$& $\mm$&$($&$W$,&$M_I$,&$<_T$&$)$
\end{tabular}
}}

In the dynamic setting, new documents are added to the corpus, and some of the old ones are removed. Simultaneously multiple users are querying on the index. Usually updates are conducted by the server, and can be easily wrapped in one write transaction. In addition, adding one document means a large set of term-document relations added to the database, and we want a whole document is combined into the database \emph{atomically}, i.e., the queries will never read a partially updated document in the database. 
The correctness would be supported by the functional tree structure. 
Assume we are adding a new document $d$ with a list of terms $\{t_i\}$ each with weight $w_i$ into the current outer tree $T$. We first build an outer-tree structure $T'$ based on all mappings $t_i\mapsto (d\mapsto w_i)$. Then we take a \texttt{union} on this tree $T'$ and the current corpus tree $T$, and whenever duplicate keys (terms) appear, we take a \texttt{union} on their values. This means that if a term $t_i\in d$ has already appeared in the current corpus $T$, the posting lists of the same term will be combined. The PAM library supports such \texttt{union} function that applies a user-specified binary operation to combine values when duplicates appear. This is done by the join-based \texttt{union} algorithms \cite{blelloch2016just}, which also runs in parallel.
\hide{
\begin{tabular}{l}
$f(\{(d_i,w_i)\})$ = $M_I.$\texttt{build\_reduce}$((d,w_i),+)$\\
$m'=M_O.$\texttt{build\_reduce}$(\{(t_i,(d,w_i)), f\})$\\
$m=M_O.$\texttt{union}$(m,m',M_I.$\texttt{union}$)$\\
\end{tabular}
We use $\{s_i\}$ to denote a sequence of elements $s_i$. In this process, all mappings of $t_i$ to $d$ are first built into an outer tree $m'$. If a term appears multiple times in the document, their weights are combined (the $f$ function taken by \texttt{build\_reduce}). Then we merge $m'$ into $m$ by \texttt{union}. If a term in $t'$ has already appeared in $t$, the posted-list will be combined using the \texttt{union} function on the inner map ($M_I.$\texttt{union}). The deletion of a file can be done in a similar way, in which the combine function in the last \texttt{union} is a $M_I.$\texttt{difference} instead of a \texttt{union}.
}

We test ``and''-queries, which means each query takes two terms and return the top-10 ranked documents in which both terms appear. We carefully choose the query terms such that the output is reasonably valid. The query is done by first read the posted-list of both terms, and take an \texttt{intersection} on them. Because of persistence, the two posting lists are just snapshots of the current database, and hence each query will not affect any other queries nor the update by the writer.

\textbf{Experimental Results.} The throughput numbers of using PAM to build or run only queries on the inverted index have been shown in \cite{pam}, and in this paper our experiments aim at showing that simultaneous updates and queries does not have much overhead comparing to running them separately. we use the publicly available Wikipedia database~\cite{Wikipedia} (dumped on Oct. 1, 2016) consisting of 8.13 million documents. We use the same pre-processing approach as in~\cite{pam}.

We first build a tree with $1.6\times 10^9$ word-doc pairs. We use different number of threads to generate queries, and the rest are used for doing updates. We note that the thread allocation for running query/update ratio depends on the scheduler. We do not use versioning or GC. We run both update and query simultaneously in 30 seconds, and record the throughput for each. We then test the same number of updates or queries running separately using all available threads (144 of them). Both update and query run in parallel---not only multiple queries run in parallel, but each single query is also parallel (using parallel intersection algorithm). The update uses a parallel union algorithm.
We report the time for running them separately as $T_u$ (purely update) and $T_q$ (purely query). Numbers are shown in Table \ref{tab:indexexp}.
As we use more threads to generate queries, the update ratio gets lower. This is because the sub-tasks in queries are generated more frequently, hence is more likely to be stolen.
In conclusion, the total time of running them almost add up to 30 seconds, which is the time running them in parallel together.

In practice, the ratio of queries running on such search engines should be much more than the updates. In this case, our experiments show that adding a single writer to update the database does not cause much overhead in running time, and the queries and gradually get the newly-added documents.

\begin{table}[!h!t]
  \centering
    \begin{tabular}{|r|r|r|r|r|}
\hline
    \multicolumn{1}{|c|}{$\bf p$} & \multicolumn{1}{c|}{$\bf T_u$} & \multicolumn{1}{c|}{$\bf T_q$} & \multicolumn{1}{c|}{$\bf T_u+T_q$} & \multicolumn{1}{c|}{$\bf T_{u+q}$} \\
\hline
    10    & 13.4  & 17.3  & 30.7  & 30 \\
\hline
    20    & 8.22  & 21.6  & 29.82 & 30 \\
\hline
    40    & 4.18  & 25.1  & 29.28 & 30 \\
\hline
    80    & 1.82  & 27    & 28.82 & 30 \\
\hline
    \end{tabular}%
  \label{tab:indexexp}%
    \caption{The running time (seconds) on the inverted index application. $T_{u+q}$ denote the time for conducting updates and queries simultaneously, using $p$ threads generating queries. We set $T_{u+q}$ to be 30s. We then record the number of updates and queries finished running, and test the same number of updates/queries separately on the initial corpus. When testing separately we use all 144 threads.}
\end{table}%

\hide{
\subsection{Summary}
Our experiments evaluated the two main aspects of our approach: the VM algorithm, and the functional data structures.
Experiments show that our functional tree structure using parallel bulk update is more efficient than many concurrent data structures, for various workloads of mixed reads and updates.
For the VM algorithm, our experiments show that the overhead caused by introducing our wait-free VM algorithm with GC is reasonably small. Comparing to the other two algorithms (HP and EP), our \PSWF{} algorithm achieves comparable queries throughput, but is faster in updates. In terms of GC, our algorithm collects out-of-date versions fairly timely (and precisely) and leaves much less garbage behind than the other two algorithms.
}

	\section{Related Work}\label{sec:related}

Multiversioning has been studied extensively since the 70s
~\cite{Reed78,BG83,papadimitriou1984concurrency}.  However, most
previous protocols, like multiversion timestamp ordering (MVTO)
\cite{Reed78} and read-only multiversion (ROMV)
~\cite{Papadimitriou1986,Weikum01} are time-stamp based, maintaining
version lists for every object, which are traversed to find the object
with the proper timestamp.  This approach inherently delays user code
since version lists can be long.  It also complicates garbage
collection.  Kumar et al.~\cite{Kumar14} revisit the MVTO protocol and
develop a concrete algorithm with GC that has similar properties to
ours if the GC is applied frequently enough.  However this requires
scanning whole version lists for objects and requires locks.  Also in
their algorithm the writer can still delay readers and the readers can
abort the writer.  As far as we know no work based on multiversioning
with version lists has shown bounds on time or space.

Perelman, Fan and Keidar~\cite{perelman2010maintaining} showed
resource bounds for multiversion protocols.  They define the notion of
MV-permissiveness, which means that only write transactions abort (or
restart), and only if they conflict.  They also define useless prefix
(UP) GC, which is similar but slightly weaker than our notion of
precise GC (it only collects proper prefixes of the versions).  They
describe an algorithm that is MV-permissive and satisfies UP GC.  They
do not give any time bounds---the delay could take time that is a
function of data structure size and number of processes, even when
there is a single writer, since the approach is based on copying an old
value to all previous active versions.

Beyond RCU~\cite{McKenney98}, the read-log-update (RLU) protocol also
supports two versions such that readers can read an old version, while
the writer updates the current version~\cite{Matveev15}. The RLU allows
readers to see the currently updated version, but still blocks before
the next version can be updated until all processes reach a quiscent
period.  Attiya and Hillel~\cite{Attiya11} suggest a similar idea
that allows readers to proceed while blocking writers (even a single
writer).

Path-copying is a default implementation in functional languages, where data cannot be overwritten \cite{okasaki1999purely}.
Similar techniques have been used for maintaining multiversion B-tree or B+tree structures or their variants \cite{minuet,becker1996asymptotically}, and is used in real-world database systems like LMDB \cite{LMDB}, CouchDB \cite{couchdb}, Hyder \cite{hyder} and InnoB \cite{innodb}, as
well as many file-systems \cite{Rodeh2013BTRFSTL,chutani1992episode,bonwick2003zettabyte,craig2003metadata,hitz1994file}.

Some techniques in our algorithm can also be found in wait-free
universal construction algorithms \cite{fatourou2011highly,
  herlihy1990methodology,herlihy1993methodology}. More details can be found in the full version of this paper.

\section{Acknowledgement}
This work was supported in part by NSF grants CCF-1408940, CCF-1533858, and CCF-1629444.


	

	


\ifappendix
\appendix
\section{Properties of the Version Maintenance Problem}\label{sec:vm-prop}
\putmaybeappendix{prop-proof}



\section{Wait-free Version Maintenance Proofs}\label{sec:waitfree-proof}

\subsection{Proof of Correctness}\label{sec:waitfree-correctness}
\putmaybeappendix{waitfree-correctness}

\subsection{Proof of Time and Contention Bounds}\label{sec:waitfree-time}
\putmaybeappendix{waitfree-time}

\section{Proof of Correct \collect{} Function}
\putmaybeappendix{collect-correctness}

\section{Proof for Single-writer Concurrency}
\subsection{Proof of Serializability}
\label{app:theorem1}
We first prove the following Lemma.
\putmaybeappendix{serial-proof}
\subsection{Proof of Safe and Precise Garbage Collection}
\label{app:theorem2}
We now show that the \collect{} algorithm is correct. First we prove that it satisfies the first part of Definition \ref{def:collect}.

\begin{lemma} \label{lem:collect-1}
Let $u$ be a shared \tuple{}. For any shared \tuple{} $w$, let $V_w$ be the set of versions that $w$ belongs to. If a \collect{} operation has terminated for each version in $V_u$, then $u$ has been freed.
\end{lemma}

\begin{proof}
Fix an execution history and a configuration $C$. Consider the set $G$ of all shared \tuple{}s $w$ such that for each version $v \in V_w$, a \collect(v){} operation has terminated. 
It suffices to show that for each \tuple{} in $G$, there is a \collect{} operation that frees the tuple and terminates before $C$.

First, we show that no local \tuple{s} can affect the \tuple{s} of $G$. To see this, fix a \tuple{} $u \in G$. We want to show that there cannot be any pointers to $u$ from local \tuple{s}, and thus that its reference count cannot be affected by local \tuple{}s. Assume by contradiction that there is a local \tuple{} $\ell$ that is pointing to $u$ in configuration $C$. Note that only write transactions ever create \tuple{s}, and that the writer cleans up local \tuple{s} in its \texttt{output} \op{}, and therefore never leaves any local \tuple{s} or effect on the reference counts of shared \tuple{s} after returning.
Therefore, $\ell$ must have been created by a write transaction $t$ that is currently running user code. For $t$ to be able to create a \tuple{} that points to $u$, there are two cases: (1) $u$ must be a part of the version that $t$ commits, or (2) $u$ must be reachable from the version that $t$ acquired.
Note that in the first case, $u$ is not a shared \tuple{} itself, since it has been created by a transaction that has not yet finished its user code.
For the second case, recall that for $u$ to be in $G$, all versions that $u$ belongs to must have been collected. However, $u$ belongs to $V(t)$, and since $t$ is running user code, $V(t)$ is live at $C$, and therefore cannot have been collected yet. This contradicts the definition of $G$.
Therefore, $\ell$ cannot exist.

Notice that $G$ forms a DAG. Furthermore, for each \tuple{} $w \in G$, $G$ contains every shared \tuple{} that points to $w$. This is because a \tuple{} belongs to all of the versions that its parent belongs to. Therefore we can proceed by structural induction on $G$.

For the base of the induction, we prove that each of the roots in $G$ has been freed by a completed \collect{} operation. Let $u$ be some root in $G$.
We just need to show that each increment of $u$'s reference count has a completed \collect($u$) operation corresponding to it. We've already shown that there are no outstanding increments from local \tuple{s} affecting $u$.
This also holds for increments by \outputt($u$) operations because all of the versions that $u$ belongs to have already been collected. Since $u$ is a root, its reference count is not incremented anywhere else, so one of the completed \collect($u$) operation sets the reference count of $u$ to $0$ and frees $u$.

Now we prove the inductive step by fixing some \tuple{} $u$ in $G$ and assuming that all of its parents have been freed by some completed \collect{} operation. Similar to the base case, we show that each increment of $u$'s reference count has a completed \collect($u$) operation corresponding to it.
All arguments from the base case hold here, and therefore we do not need to worry about increments from local \tuple{s} or \outputt{} \op{}s.
So we just need to show that for each shared \tuple{} $w$ that point to $u$, there is also a completed \collect($u$) \op{}. By the inductive hypothesis, there is a completed \collect{} operation that frees $w$, and we can see from the code that this operation executes a \collect{} on $u$. Therefore one of the completed \collect($u$) operation sets the reference count of $u$ to $0$ and frees $u$.
By structural induction, each tuple in $G$ has been freed and this completes the proof.
\end{proof}

Next we prove that our collect algorithm satisfies the second part of Definition \ref{def:collect}.

\begin{lemma} \label{lem:collect-2}
Let $u$ be a shared \tuple{} and let $V_u$ be the set of versions that it belongs to. If a \collect{} operation has not started for some version $v \in V_u$, then $u$ has not been freed.
\end{lemma}

\begin{proof}

	Next we claim that each \collect($u$) operation corresponds to an unique increment of $u$'s reference counter.
	This can be seen by a close inspection of the code; let $c$ be a \collect($u$) call and consider two cases. Case (1): $c$ is not called from inside another \collect. That is, $u$ is the root of a version that is being collected. In that case, $c$ corresponds to the increment of $u.\mbox{\it ref}$ in the \outputt{} \op{} of the write that committed this version. Case (2): $c$ is called recursively from a \collect($u'$) \op{}. In this case, the $c$ corresponds to the increment of $u.\mbox{\it ref}$ during the creation of $u'$.

	Let $v \in V_u$ be the version for which no \collect($v$) call has been invoked.
	Since $u$ belongs to $v$, there must be a path from $v$'s version root $r$ to $u$ in the memory graph. We show by induction that no \tuple{} along that graph has been freed, thus implying that $u$ has not been freed.
	
	\textsc{Base:} Consider $v$'s root, $r$. $r.ref$ has been incremented by the \texttt{output} call of the writer that created the version $v$ and the \collect($v$) operation corresponding to this increment has not been invoked yet. Therefore the reference count of $r$ is non-zero, so it has not been freed.
	
	\textsc{Step:} Assume that the $i$th \tuple{}, $u_i$ in the path from $r$ to $u$ is not freed. We want to show that the $i+1$th \tuple{} on this path, $u_{i+1}$ has not been freed either.
	Consider the \texttt{\tuple} \op{} that made $u_i$ the parent of $u_{i+1}$ in the memory graph. That \op{} incremented $u_{i+1}$'s reference count by $1$ and the \collect($v$) operation corresponding to this increment has not been invoked yet because $u_i$ has not been freed.
	Thus, $u_{i+1}$'s reference count is greater than $0$, and therefore it cannot have been freed.
\end{proof}

Finally, we prove that our \collect{} algorithm is efficient.

\begin{lemma} \label{lem:collect-3}
A \collect{} operation takes $O(S+1)$ time where $S$ is the number of \tuple{}s that were freed by the operation.
\end{lemma}

\begin{proof}
Not counting the recursive calls, each \collect{} operation needs a constant time. Each time a \tuple{} is freed, a \collect{} operation is called on each of its $l$ children. Therefore, the total number of \collect{} operations spawned by a \collect{} operation $C$ is $l\times S$, where $S$ is the number of \tuple{}s that were freed by $C$. Since $l$ is constant, $C$ has $O(S+1)$ time complexity in total.
\end{proof}

Together, Lemmas \ref{lem:collect-1}, \ref{lem:collect-2} and \ref{lem:collect-3} imply Theorem \ref{thm:collect-correct}.
\putmaybeappendix{discussion}

\section{Batching} 
\label{sec:batching}
\hide{
\lstset{basicstyle=\footnotesize\ttfamily, tabsize=2, escapeinside={@}{@}}
\lstset{literate={<}{{$\langle$}}1  {>}{{$\rangle$}}1}
\lstset{language=C, morekeywords={CAS,commit,empty,local,job,taken,entry,GOTO}}
\lstset{xleftmargin=5.0ex, numbers=left, numberblanklines=false, frame=single}
\begin{figure}
	\caption{Batching Algorithm}
\begin{lstlisting}[countblanklines=false]
int head[P], tail[P], offset[P];
void update(operation k, size_t id) {
	buffer[id].push_back(k); }

void single_writer() {
	while (true) {
		while (not enough operations
			&& not a long time) {}
		int start[P], end[P], offset[P];
		for (int i = 0; i < P; i++) {
			start[i] = head[i]; end[i] = tail[i];
			int block = end[i]-start[i];
			if (block > max_op_each) {
				end[i] = start[i]+max_op_each;
				block = max_op_each;
			}
			offset[i+1] = offset[i]+block;
			head[i] = end[i]; }
		int m = offset[P];
		entry_type* a = new entry_type[m];
        		
		par_for (int_t i=0; i@<@P; i++) {
			par_for(int j=start[i]; j@<@end[i]; j++) {
				int ind = offset[i]+j-start[i];
				a[ind] = buffer[i][k]; } }
		d = commit_batch(d, a, m); }
}	
\end{lstlisting}
\vspace{-.3in}
\label{alg:batching}
\end{figure}}


For multi-writer settings, our approach allows abort to avoid write-write conflict.
In order to allow no-abort, an option is to use batching with a global single writer transaction.
This transaction is responsible for collecting concurrent updates, and committing the whole batch
atomically using \set{}. This writer transaction can itself run in parallel.
Similar ideas appear in flat combining \cite{hendler2010flat}, which is known to be efficient in practice.

In this paper, we use this technique as part of the user-code of the writer to commit a batch of write operations.
We use a simple strategy, where each process is allocated a buffer array with a head and a tail index.
Each process submits all its updates to the buffer by adding them to the tail. Periodically, the writer goes over each array, assembles all operations between the current head and tail into the batch, and then moves the head index to the current tail index (plus one). There is no contention between processes because each reader only operates its own buffer at the tail, and the single writer only operate on the head index of all buffers.


The updates are then committed to the database in a batch, possibly in parallel. For example, in our experiments we use functional tree structures as the underlying data structure, and multiple inserts or updates can be done using a parallel \texttt{multi\_insert} function~\cite{pam}. At any time, no two concurrent threads can work on the same tree node. This avoids contention between writes, while utilizing multiple cores to improve throughput. Our experiments show that it is efficient in practice.
Our approach also allows each batch to be committed \emph{atomically}, since committing the new version root makes all new tuples visible atomically.
For more complicated transactions, we need to first build the dependency between transactions.

Typically, a larger batch size leads to higher throughput because of better parallelism, but at the cost of longer latency. In our implementation, we control the latency to be at the same magnitude of network latency, such that the latency waiting for a batch to finish does not dominate the cost.

We note that the batching scheme invalidates the guarantee of wait-freedom, but we will show in the experiments that it is generally fast in practice, and can be even faster than state-of-the-art concurrent data structures.

\fi
\end{document}